\documentclass[runningheads]{llncs}
\usepackage{graphicx}
\usepackage{xspace}
\usepackage{todonotes}
\usepackage{mathtools, amsfonts}
\usepackage{cite}
\usepackage{microtype}
\usepackage{caption, subcaption}
\usepackage{algorithm}
\usepackage[]{algpseudocode}
\usepackage{textcomp, mathcomp}
\usepackage{array}
\usepackage{comment}

\usepackage{hyperref}
\usepackage{bibentry}

\newcommand{\OPT}{\text{OPT}}
\newcommand{\TOP}{\text{TOP}}
\newcommand{\BOT}{\text{BOT}}
\newcommand{\MIN}{\text{MIN}}

\begin{document}
\title{Approximation algorithms and an integer program for multi-level graph spanners} 
\titlerunning{Approximation algorithms and an ILP for multi-level graph spanners}
%

\author{Reyan Ahmed\inst{1}\and
Keaton Hamm\inst{1} \and
Mohammad Javad Latifi Jebelli\inst{1} \and
Stephen Kobourov\inst{1} \and
Faryad Darabi Sahneh\inst{1} \and
Richard Spence\inst{1}}
\authorrunning{Ahmed et al.}
%
\institute{University of Arizona, USA}
\maketitle              
\begin{abstract}
Given a weighted graph $G(V,E)$ and $t \ge 1$, a subgraph $H$ is a \emph{$t$--spanner} of $G$ if the lengths of shortest paths in $G$ are preserved in $H$ up to a multiplicative factor of $t$. The \emph{subsetwise spanner} problem aims to preserve distances in $G$ for only a subset of the vertices. We generalize the minimum-cost subsetwise spanner problem to one where vertices appear on multiple levels, which we call the \emph{multi-level graph spanner} (MLGS) problem, and describe two simple heuristics. Applications of this problem include road/network building and multi-level graph visualization, especially where vertices may require different grades of service.

We formulate a 0--1 integer linear program (ILP) of size $O(|E||V|^2)$ for the more general minimum \emph{pairwise spanner problem}, which resolves an open question by Sigurd and Zachariasen on whether this problem admits a useful polynomial-size ILP. We extend this ILP formulation to the MLGS problem, and evaluate the heuristic and ILP performance on random graphs of up to 100 vertices and 500 edges.

\end{abstract}

\section{Introduction}

Given an undirected edge-weighted graph $G(V,E)$ and a real number $t \ge 1$, a subgraph $H(V,E')$ is a (multiplicative) \emph{$t$--spanner} of $G$ if the lengths of shortest paths in $G$ are preserved in $H$ up to a multiplicative factor of $t$; that is, $d_H(u,v) \le t \cdot d_G(u,v)$ for all $(u,v) \in V\times V$, where $d_G(u,v)$ is the length of the shortest path from $u$ to $v$ in $G$. We refer to $t$ as the \textit{stretch factor} of $H$. Peleg et al.~\cite{doi:10.1002/jgt.3190130114} show that determining if there exists a $t$--spanner of $G$ with $m$ or fewer edges is NP--complete. Further, it is NP--hard to approximate the (unweighted) $t$--spanner problem to within a factor of $O(\log |V|)$, even when restricted to bipartite graphs~\cite{Kortsarz2001}.

In the \emph{pairwise spanner} problem \cite{Cygan13}, distances only need to be preserved for a subset $\mathcal{P} \subseteq V \times V$ of pairs of vertices. Thus, the classical $t$--spanner problem is a special case of the pairwise spanner problem where $\mathcal{P} = V \times V$. The \emph{subsetwise spanner} problem is a special case of the pairwise spanner problem where $\mathcal{P} = S \times S$ for some $S \subset V$; that is, distances need only be preserved between vertices in $S$ \cite{Cygan13}.  The case $t=1$ is known as the \emph{pairwise distance preserver} or \emph{sourcewise distance preserver} problem, respectively\cite{coppersmith2006sparse}. The subsetwise spanner problem where $t$ is arbitrarily large is known as the \emph{Steiner tree} problem on graphs.

\subsection{Multi-level graph spanners}


In many network design problems, vertices or edges come with a natural notion of priority, grade of service, or level; see Fig.~\ref{fig:map_vs_spanner}. For example, consider the case of rebuilding a transportation infrastructure network after a natural disaster.  Following such an event, the rebuilding process may wish to prioritize connections between important buildings such as hospitals or distribution centers, making these higher level terminals, while ensuring that no person must travel an excessive distance to reach their destination. Such problems have been referred to by names such as hierarchical network design, grade of service problems, multi-level, multi-tier, and have applications in network routing and visualization.

\begin{figure}[t]
\centering
\vspace{-.4cm}\includegraphics[width=\textwidth]{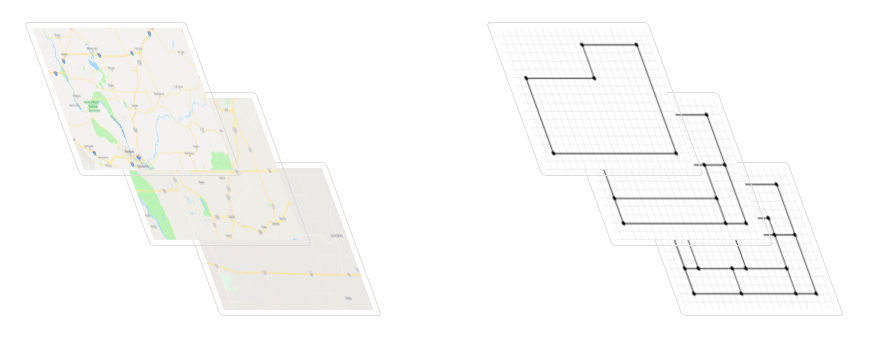}\vspace{-.2cm}
\caption{An interactive road map serves as a good analogy for the MLGS problem, where the top level graph $G_{\ell}$  represents the network of major highways, and zooming in to $G_{\ell-1}$ shows a denser network of smaller roads.
\label{fig:map_vs_spanner}}
\end{figure}

Similar to other graph problems which generalize to multiple levels or grades of service \cite{1288137}, we extend the subsetwise spanner problem to the \textit{multi-level graph spanner (MLGS)} problem: 
\begin{definition}\label{DEF:MLGS}[Multi-level graph spanner (MLGS) problem]
Given a graph $G(V,E)$ with positive edge weights $c: E \to \mathbb{R}_{+}$, a nested sequence of \textit{terminals}, $T_{\ell} \subseteq T_{\ell-1} \subseteq \ldots \subseteq T_1 \subseteq V$, and a real number $t \ge 1$, compute a minimum-cost sequence of spanners $G_{\ell} \subseteq G_{\ell-1} \subseteq \ldots \subseteq G_1$, where $G_i$ is a subsetwise $(T_i \times T_i)$--spanner for $G$ with stretch factor $t$ for $i=1,\dots,\ell$. The cost of a solution is defined as the sum of the edge weights on each graph $G_i$, i.e., $\sum_{i=1}^{\ell} \sum_{e \in E(G_i)}c_e$.
\end{definition}

We refer to $T_i$ and $G_i$ as the terminals and the graph on level $i$. A more general version of the MLGS problem can involve different stretch factors on each level or a more general definition of cost, but for now we use the same stretch factor $t$ for each level.

An equivalent formulation of the MLGS problem which we use interchangeably involves \emph{grades of service}: given $G = (V,E)$ with edge weights, and required grades of service $R: V \to \{0,1,\ldots,\ell\}$, compute a single subgraph $H \subseteq G$ with varying grades of service on the edges, with the property that for all $u,v \in V$, if $u$ and $v$ each have required a grade of service greater than or equal to $i$, then there exists a path in $H$ from $u$ to $v$ using edges with a grade of service greater than or equal to $i$, and whose length is at most $t \cdot d_G(u,v)$. Thus, $T_{\ell} = \{v \in V \mid R(v) = \ell\}$, $T_{\ell-1} = \{v \in V \mid R(v) \ge \ell-1\}$, and so on. If $y_e$ denotes the grade of edge $e$ (or the number of levels $e$ appears in), then the cost of a solution is equivalently $\sum_{e \in H} c_e y_e$, that is, edges with a higher grade of service incur a greater cost. This interpretation makes it clear that more important vertices (e.g., hubs) are connected with higher quality edges; see example instance and solution in  Fig.~\ref{fig:mlgs-example}.

\begin{figure}[t]
\centering
\vspace{-.2cm}\includegraphics[width=\textwidth]{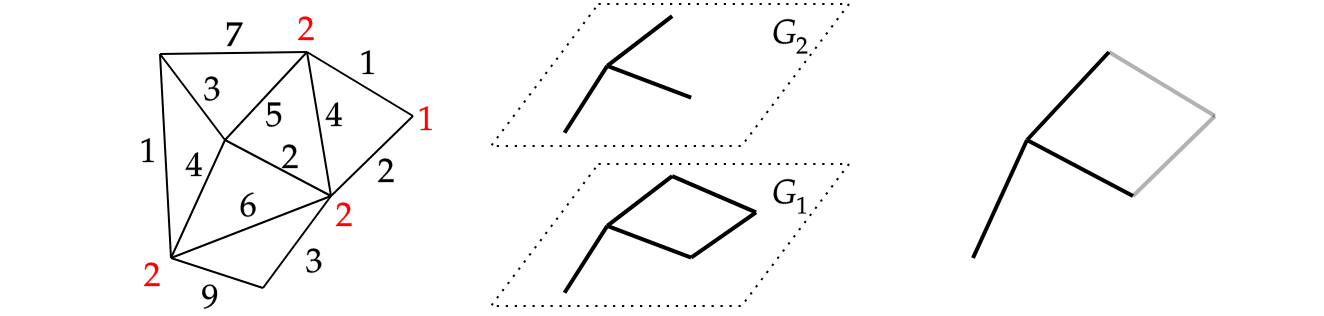}
\caption{\emph{Left:} Input graph $G$ with edge weights, $\ell = 2$, $|T_2| = 4$, $|T_1| = 3$, and $t = 3$. Required grades of service $R(v)$ are shown in red. \emph{Center:} A valid MLGS $G_2 \subseteq G_1 \subseteq G$ is shown. \emph{Right:} The equivalent solution, where dark edges $e$ have $y_e = 2$ and light edges have $y_e = 1$. The cost of this solution is $2 \times (4+2+5) + 1 \times (1+2) = 25$.}
\label{fig:mlgs-example}
\end{figure}

If $t$ is arbitrarily large, the MLGS problem reduces to the \emph{multi-level Steiner tree} (MLST) problem\cite{MLST2018}. However it is worth noting that the problem of computing or approximating spanners is significantly harder than that of computing Steiner trees, and that a Steiner tree of $G$ may be an arbitrarily poor spanner; a cycle on $|V|$ vertices with one edge removed is a possible Steiner tree of $G$, but is only a $(|V|-1)$-spanner of $G$.  The techniques used here have similarities to those used in the MLST problem, but more sophisticated methods are needed as well, including the use of approximate distance preservers and a new ILP formulation for the pairwise spanner problem.




\subsection{Related work}
Spanners and variants thereof have been studied for at least three decades, so we focus on results relating to pairwise or subsetwise spanners. Alth\"{o}fer et al.~\cite{Alth90} provide a simple greedy algorithm that constructs a multiplicative $r$--spanner given a graph $G$ and a real number $r > 0$. The greedy algorithm sorts edges in $E$ by nondecreasing weight, then for each $e = \{u,v\} \in E$, computes the shortest path $P(u,v)$ from $u$ to $v$ in the current spanner, and adds the edge to the spanner if the weight of $P(u,v)$ is greater than $r \cdot c_e$. The resulting subgraph $H$ is a $r$--spanner for $G$. The main result of \cite{Alth90} is that, given a weighted graph $G$ and $t \ge 1$, there is a greedy $(2t+1)$--spanner $H$ containing at most $n \lceil n^{1/t}\rceil$ edges, and whose weight is at most $w(MST(G)) (1 + \frac{n}{2t})$ where $w(MST(G))$ denotes the weight of a minimum spanning tree of $G$.

Sigurd and Zachariasen~\cite{sigrd04} present an ILP formulation for the minimum-weight pairwise spanner problem (see Section~\ref{section:ilp}), and show that the greedy algorithm~\cite{Alth90} performs well on  sparse graphs of up to 64 vertices. \'{A}lvarez-Miranda and Sinnl \cite{Alvarez-Miranda2018} present a mixed ILP formulation for the tree $t^*$--spanner problem, which asks for a spanning tree of a graph $G$ with the smallest stretch factor $t^*$.


Dinitz et al.~\cite{Dinitz11} provide a flow-based linear programming (LP) relaxation to approximate the directed spanner problem. Their LP formulation is similar to that in~\cite{sigrd04}; however, they provide an approximation algorithm which relaxes their ILP, whereas the previous formulation was used to compute spanners to optimality. Additionally, the LP formulation applies to graphs of unit edge cost; they later take care of it in their rounding algorithm by solving a shortest path arborescence problem. They provide a $\tilde{O}(n^{\frac{2}{3}})$--approximation algorithm for the directed $k$--spanner problem for $k \geq 1$, which is the first sublinear approximation algorithm for arbitrary edge lengths. 
Bhattacharyya et al.~\cite{Bhattacharyya12} provide a slightly different formulation to approximate $t$--spanners as well as other variations of this problem. 
They provide a polynomial time $O((n \log n)^{1-\frac{1}{k}})$--approximation algorithm for the directed $k$--spanner problem. Berman et al.~\cite{BERMAN13} provide an alternative randomized LP rounding schemes that lead to better approximation ratios. They improved the approximation ratio to $O(\sqrt{n} \log n)$ where the approximation ratio of the algorithm provided by Dinitz et al.~\cite{Dinitz11} was $O(n^\frac{2}{3})$. They have also improved the approximation ratio for the important special case of directed 3--spanners with unit edge lengths.

There are several results on multi-level or grade-of-service Steiner trees, e.g.,  \cite{MLST2018,Balakrishnan1994,1288137,Chuzhoy2008,mirchandani1996MTT}, while  multi-level spanner problems have not been studied yet.

\section{Approximation algorithms for MLGS}
\label{section:approximation_algorithms}

Here, we assume an oracle subroutine that computes an optimal $(S\times S)$--spanner, given a graph $G$, subset $S \subseteq V$, and $t$. The intent is to determine if approximating MLGS is significantly harder than the subsetwise spanner problem. We formulate simple bottom-up and top-down approaches for the MLGS problem.

\subsection{Oracle bottom-up approach} \label{subsection:oracle-bu}
The approach is as follows: compute a minimum subsetwise $(T_1 \times T_1)$--spanner of $G$ with stretch factor $t$. This immediately induces a feasible solution to the MLGS problem, as one can simply copy each edge from the spanner to every level (or, in terms of grades of service, assign grade $\ell$ to each spanner edge). We  then prune edges that are not needed on higher levels. It is easy to show that the solution returned has cost no worse than $\ell$ times the cost of the optimal solution. Let $\OPT$ denote the cost of the optimal MLGS $G_{\ell}^* \subseteq G_{\ell-1}^* \subseteq \ldots \subseteq G_1^*$ for a graph $G$. Let $\MIN_i$ denote the cost of a minimum subsetwise $(T_i \times T_i)$--spanner for level $i$ with stretch $t$, and let $\BOT$ denote the cost computed by the bottom-up approach. If no pruning is done, then $\BOT = \ell \MIN_1$.

\begin{theorem}
The oracle bottom-up algorithm described above yields a solution that satisfies $\BOT \le \ell \cdot \OPT$.
\end{theorem}
\begin{proof}
We know $\MIN_1 \le \OPT$, since the lowest-level graph $G_1^*$ is a $(T_1 \times T_1)$--spanner whose cost is at least $\MIN_1$. Further, we have $\BOT = \ell \MIN_1$ if no pruning is done. Then $\MIN_1 \le \OPT \le \BOT = \ell \cdot \MIN_1$, so $\BOT \le \ell \cdot \OPT$. \hfill $\square$
\end{proof}

The ratio of $\ell$ is asymptotically tight; an example can be constructed by letting $G$ be a cycle containing $t$ vertices and all edges of cost 1. Let two adjacent vertices in $G$ appear in  $T_{\ell}$, while all vertices appear in $T_1$, as shown in Figure \ref{fig:topbot-tightness}. As $t \to \infty$, the ratio $\frac{\BOT}{\OPT}$ approaches $\ell$. Note that in this example, no edges can be pruned without violating the $t$--spanner requirement.


\begin{figure} [t]
    \centering
    \includegraphics[width=4in]{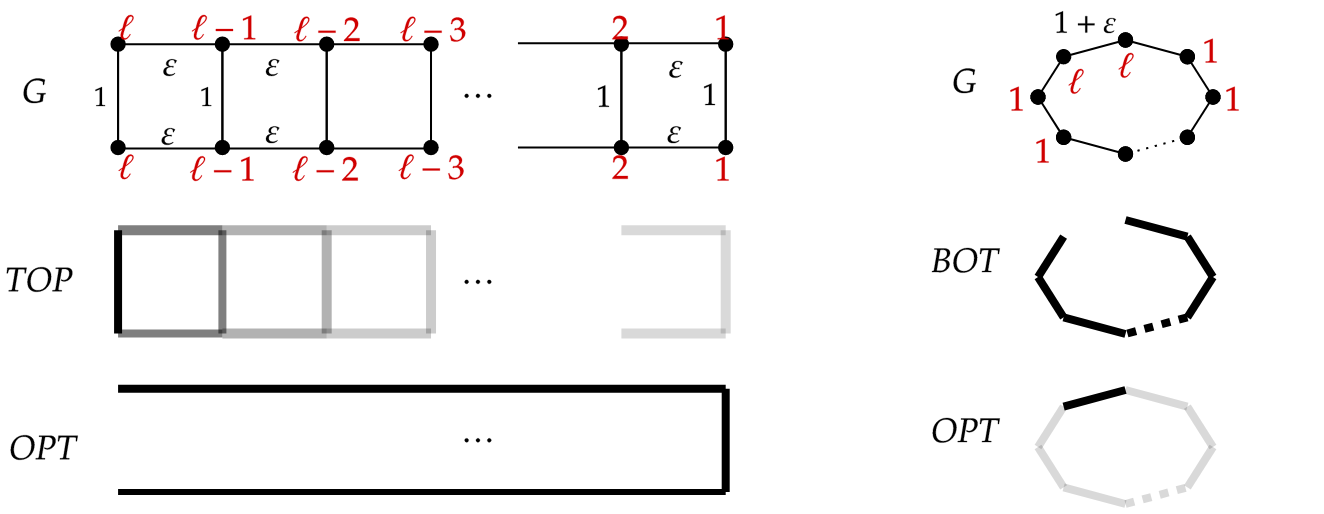}
    \caption{\emph{Left:} Tightness example of the top-down approach. Consider the lattice graph $G$ with pairs of vertices of  grade $\ell$ ($|T_{\ell}| = 2$), $\ell-1$, and so on. The edge connecting the two vertices of grade $i$ has weight $1$, and all other edges have weight $\varepsilon$, where $0 < \varepsilon \ll 1$. Set $t = 2$. The top-down solution (middle) has cost $\TOP \approx \ell + (\ell-1) + \ldots + 1  = \frac{\ell(\ell+1)}{2}$, while the optimal solution (bottom) has cost $\OPT \approx \ell$.
    \emph{Right:} Tightness example of the bottom-up approach. Consider a cycle $G$ containing two adjacent vertices of grade $\ell$, and the remaining vertices of grade 1. The edge connecting the two vertices of grade $\ell$ is $1+\varepsilon$, while the remaining edges have weight 1. Setting $t = |E|$ yields  $\BOT = \ell|E|$ while $\OPT = (1+\varepsilon)\ell + 1(|E|-1) \approx |E|+\ell$.}
    \label{fig:topbot-tightness}
\end{figure}

We give a simple heuristic that attempts to ``prune'' unneeded edges without violating the $t$--spanner requirement. Note that any pruning strategy may not prune any edges, as a worst case example (Figure \ref{fig:topbot-tightness}) cannot be pruned. Let $G_1$ be the $(T_1 \times T_1)$--spanner computed by the bottom-up approach. To compute a $(T_2 \times T_2)$--spanner $G_2$ using the edges from $G_1$, we can compute a distance preserver of $G_1$ over terminals in $T_2$. One simple strategy is to use shortest paths as explained below.




Even more efficient pruning is possible through the {\em distant preserver} literature \cite{coppersmith2006sparse,bodwin2017linear}. A well-known result of distant preservers is due to the following theorem:

\begin{theorem}[\cite{coppersmith2006sparse}] \label{Elkin-preserver}
Given $G=(V,E)$ with $|V|=n$, and $P \subset {V \choose 2}$, there exists a subgraph $G'$ with $O(n+\sqrt{n} |P|)$ edges such that for all $(u,v) \in P$ we have $d_{G'}(u,v) = d_G(u,v)$.
\end{theorem}

The above theorem hints at a {\em sparse} construction of $G_2$ simply by letting $P=T_2\times T_2$. Given $G_1$, let $G_i$ be a distance preserver of $G_{i-1}$ over the terminals $T_i$, for all $i = 2, \ldots, \ell$. An example is to let $G_2$ be the union of all shortest paths (in $G_1$) over vertices $v,w \in G_2$. The result is clearly a feasible solution to the MLGS problem, as the shortest paths are preserved exactly from $G_1$, so each $G_i$ is a $(T_i \times T_i)$--spanner of $G$ with stretch factor $t$.




\subsection{Oracle top-down approach} \label{subsection:oracle-td}
A simple top-down heuristic that computes a solution is as follows: let $G_{\ell}$ be the minimum-cost $(T_{\ell} \times T_{\ell})$--spanner over terminals $T_{\ell}$ with stretch factor $t$, and cost $\MIN_{\ell}$. Then compute a minimum cost $(T_{\ell-1} \times T_{\ell-1})$--spanner over $T_{\ell-1}$, and let $G_{\ell-1}$ be the union of this spanner and $G_{\ell}$. Continue this process, where $G_i$ is the union of the minimum cost $(T_i \times T_i)$--spanner and $G_{i+1}$. Clearly, this produces a feasible solution to the MLGS problem.






The solution returned by this approach, with cost denoted by $\TOP$, is not worse than $\frac{\ell+1}{2}$ times the optimal. 
Define $\MIN_i$ and $\OPT$ as before. Define $\OPT_i$ to be the cost of edges on level $i$ but not level $i+1$ in the optimal MLGS solution, so that $\OPT = \ell \OPT_{\ell} + (\ell-1)\OPT_{\ell-1} + \ldots + \OPT_1$. Define $\TOP_i$ analogously.

\begin{theorem} \label{lemma:mlgs-td}
The oracle top-down algorithm described above yields an approximation that satisfies the following:
\begin{enumerate}
\item[(i)] $\TOP_{\ell} \le \OPT_{\ell}, $ \label{td-i}

\medskip
\item[(ii)] $\TOP_i \le \OPT_i + \OPT_{i+1} + \ldots + \OPT_{\ell},\quad i=1,\dots,\ell-1,$ \label{td-ii}

\medskip
\item[(iii)] $\TOP\le\frac{\ell+1}{2}\OPT.$
\end{enumerate}
\end{theorem}
\begin{proof}
Inequality (i) is true by definition, as we compute an optimal $(T_{\ell} \times T_{\ell})$--spanner whose cost is $\TOP_{\ell}$, while $\OPT_{\ell}$ is the cost of some $(T_{\ell} \times T_{\ell})$-spanner. For (ii), note that $\TOP_i \le \MIN_i$, with equality when the minimum-cost $(T_i \times T_i)$--spanner and $G_{i+1}$ are disjoint. The spanner of cost $\OPT_i + \OPT_{i+1} + \ldots + \OPT_{\ell}$ is a feasible $(T_i \times T_i)$--spanner, so $\MIN_i \le \OPT_i + \ldots + \OPT_{\ell}$, which shows (ii).

To show (iii), note that (i) and (ii) imply
\begin{align*}
\TOP &= \ell \TOP_{\ell} + (\ell-1)\TOP_{\ell-1} + \ldots + \TOP_1 \\
&\le \ell \OPT_{\ell} + (\ell-1)(\OPT_{\ell-1} + \OPT_{\ell}) + \ldots + (\OPT_1 + \OPT_2 + \ldots + \OPT_{\ell}) \\
&= \frac{\ell(\ell+1)}{2}\OPT_{\ell} + \frac{(\ell-1)\ell}{2}\OPT_{\ell-1} + \ldots + \frac{1 \cdot 2}{2}\OPT_1 \\
&\le \frac{\ell+1}{2}\OPT,
\end{align*}
as by definition $\OPT = \ell \OPT_{\ell} + (\ell-1)\OPT_{\ell-1} + \ldots + \OPT_1$. \hfill$\square$
\end{proof}

The ratio $\frac{\ell+1}{2}$ is tight as illustrated in Figure \ref{fig:topbot-tightness}, left.

\subsection{Combining top-down and bottom-up}
\label{subsection:other-heuristics}
Again, assume we have access to an oracle that computes a minimum weight $(S \times S)$--spanner of an input graph $G$ with given stretch factor $t$. A simple combined method, similar to \cite{MLST2018}, is to run the top-down and bottom-up approaches for the MLGS problem, and take the solution with minimum cost. This has a slightly better approximation ratio than either of the two approaches.

\begin{theorem} \label{thm:tdbu}
The solution whose cost is $\min(\TOP, \BOT)$ is not worse than $\dfrac{\ell+2}{3}$ times the cost $\OPT$ of the optimal MLGS.
\end{theorem}
The proof is given in Appendix \ref{apdx:tdbu}.

\subsection{Heuristic subsetwise spanners}
\label{heuristic_subsetwise_spanners}
 So far, we have assumed that we have access to an optimal subsetwise spanner given by an oracle. Here we propose a heuristic algorithm to compute subsetwise spanner. The key idea is to apply the greedy spanner to an auxiliary complete graph with terminals as its vertices and the shortest distance between terminals as edge weights. Then, we apply the distance preserver discussed in Theorem \ref{Elkin-preserver} to construct a subsetwise spanner.


\begin{theorem}
Given graph $G(V,E)$, stretch factor $t\ge1$, and subset $T\subset V$, there exists a $(T\times T)$--spanner for $G$ with stretch factor $t$ and $O(n+\sqrt{n}|T|^{1+\frac{2}{t+1}})$ edges.
\end{theorem}

\begin{proof}

The spanner may be constructed as follows:
 
 \begin{enumerate}
  \item  Construct the  terminal complete graph $\bar{G}$ whose vertices are $\bar{V}:= T$, such that the weight of each edge $\{u,v\}$ in $\bar{G}$ is the length of the shortest path connecting them in $G$, i.e., $w(u,v) = d_G(u,v)$.
  \item Construct a greedy $t-$spanner $\bar{H}(\bar{V},\bar{E}')$ of $\bar{G}$. According to \cite{Alth90}, this graph has $|T|^{1+\frac{2}{t+1}}$ edges.  Let $P=\bar{E}'$.

  \item Apply Theorem \ref{Elkin-preserver} to obtain a subgraph $H$ of $G$ such that for all $(u,v) \in P$ we have $d_H(u,v) = d_G(u,v) $. Therefore, for arbitrary $u,v \in T$ we get $d_H(u,v) \leq t \, d_G(u,v)$. 
  \item Finally, let shortest-path$(u,v)$ be the collection of edges in the shortest path from $u$ to $v$ in $H$, and 
  \[
  E = \bigcup\limits_{(u,v) \in P} \{ e \in E \,|\, e \in \textrm{shortest-path}(u,v). \, \}
  \]
 \end{enumerate}
 
According to Theorem \ref{Elkin-preserver}, the number of edges in the constructed spanner $H(V,E)$ is $O(n+\sqrt{n}|P|)=O\left(n+\sqrt{n}|T|^{1+\frac{2}{t+1}}\right)$.\hfill$\square$
\end{proof}

Hence, we may replace the oracle in the top-down and bottom-up approaches (Sections \ref{subsection:oracle-bu}-\ref{subsection:oracle-td}) with the above heuristic; we call the resulting algorithms  heuristic top-down and heuristic bottom-up. We analyze the performance of all algorithms 
on several types of graphs.

Incorporating the heuristic subsetwise spanner in our top-down and bottom up heuristics has two implications. First, the size of the final MLGS is dominated by the size of the spanner at the bottom level, i.e., $O(n+\sqrt{n}|P|)=O\left(n+\sqrt{n}|T_1|^{1+\frac{2}{t+1}}\right)$. Second, since the greedy spanner algorithm used in the above subsetwise spanner can produce spanners that are $O(n)$ more costly than the optimal solution, the same  applies to the subsetwise spanner. Our experimental results, however, indicate that the heuristic approaches are very close to the optimal solutions obtained via our ILP.

\section{Integer linear programming (ILP) formulations}
\label{section:ilp}
We describe the original ILP formulation for the pairwise spanner problem~\cite{sigrd04}. 
Let $K = \{(u_i, v_i)\}\subset V\times V$ be the set of vertex pairs; recall that the $t$--spanner problem is a special case where $K = V \times V$. Here we will use unordered pairs of distinct vertices, so in the $t$--spanner problem we have $|K| = \binom{|V|}{2}$ instead of $|V|^2$. 
This ILP formulation uses \emph{paths} as decision variables. Given $(u, v) \in K$, denote by $P_{uv}$ the set of paths from $u$ to $v$ of cost no more than $t \cdot d_G(u,v)$, and denote by $P$ the union of all such paths, i.e., $P = \bigcup_{(u,v) \in K} P_{uv}$. Given a path $p \in P$ and edge $e \in E$, let $\delta_p^e = 1$ if $e$ is on path $p$, and $0$ otherwise. Let $x_e = 1$ if $e$ is an edge in the pairwise spanner $H$, and $0$ otherwise. Given $p \in P$, let $y_p = 1$ if path $p$ is in the spanner, and zero otherwise. An ILP formulation for the pairwise spanner problem is given below.

\begin{align}
    \allowdisplaybreaks
    \text{Minimize} \sum_{e \in E} c_ex_e \text{ subject to}\\
    \sum_{p \in P_{uv}} y_p \delta_p^e &\le x_e & \forall e \in E; \forall (u,v) \in K \label{eqn:ilp-sigurd-1}\\
    \sum_{p \in P_{uv}} y_p &\ge 1 & \forall (u,v) \in K \label{eqn:ilp-sigurd-2}\\
    x_e &\in \{0,1\} & \forall e \in E \\
    y_p &\in \{0,1\} & \forall p \in P
\end{align}

Constraint (\ref{eqn:ilp-sigurd-2}) ensures that for each pair $(u,v) \in K$, at least one $t$--spanner path is selected, and constraint (\ref{eqn:ilp-sigurd-1}) enforces that on the selected $u$-$v$ path, every edge along the path appears in the spanner. The main drawback of this ILP is that the number of path variables is exponential in the size of the graph. The authors use delayed column generation by starting with a subset $P' \subset P$ of paths, with the starting condition that for each $(u,v) \in K$, at least one $t$--spanner path in $P_{uv}$ is in $P'$. The authors leave as an open question whether this problem admits a useful polynomial-size ILP. 

We introduce a 0-1 ILP formulation for the pairwise $t$--spanner problem based on multicommodity flow, which uses $O(|E||K|)$ variables and constraints, where $|K| = O(|V|^2)$. Define $t$, $c_e$, $d_G(u,v)$, $K$, and $x_e$ as before. Note that $d_G(u,v)$ can be computed in advance, using any all-pairs shortest path (APSP) method.

Direct the graph by replacing each edge $e = \{u,v\}$ with two edges $(u,v)$ and $(v,u)$ of weight $c_e$. Let $E'$ be the set of all directed edges, i.e., $|E'| = 2|E|$. Given $(i,j) \in E'$, and an unordered pair of vertices $(u,v) \in K$, let $x_{(i,j)}^{uv} = 1$ if edge $(i,j)$ is included in the selected $u$-$v$ path in the spanner $H$, and 0 otherwise. This definition of path variables is similar to that by \'{A}lvarez-Miranda and Sinnl~\cite{Alvarez-Miranda2018} for the tree $t^*$--spanner problem. 
We select a total order of all vertices so that the path constraints (\ref{eqn:ilp-2})--(\ref{eqn:ilp-3}) are well-defined. This induces $2|E||K|$ binary variables, or $2|E|\binom{|V|}{2} = 2|E||V|(|V|-1)$ variables in the standard $t$--spanner problem. Note that if $u$ and $v$ are connected by multiple paths in $H$ of length $\le t \cdot d_G(u,v)$, we need only set $x_{(i,j)}^{uv} = 1$ for edges along some path. Given $v \in V$, let $In(v)$ and $Out(v)$ denote the set of incoming and outgoing edges for $v$ in $E'$. In \eqref{eqn:ilp-1}--\eqref{eqn:ilp-5} we assume $u < v$ in the total order, so spanner paths are from $u$ to $v$. An ILP formulation for the pairwise spanner problem is as follows.

\begin{align}
    \allowdisplaybreaks
    \text{Minimize} \sum_{e \in E} c_ex_e \text{ subject to} \label{eqn:ilp-obj}\\
    \sum_{(i,j) \in E'} x_{(i,j)}^{uv} c_e &\le t \cdot d_G(u,v) & \hspace{-10pt}\forall (u,v) \in K; e = \{i,j\} \label{eqn:ilp-1}\\
    \sum_{(i,j) \in Out(i)} x_{(i,j)}^{uv} - \sum_{(j,i) \in In(i)} x_{(j,i)}^{uv} &= \begin{cases}
        1 & i = u \\
        -1 & i = v \\
        0 & \text{else}
    \end{cases} & \hspace{-10pt} \forall (u,v) \in K; \forall i \in V \label{eqn:ilp-2} \\
    \sum_{(i,j) \in Out(i)} x_{(i,j)}^{uv} & \le 1 &\forall (u,v) \in K; \forall i \in V \label{eqn:ilp-3}\\
    x_{(i,j)}^{uv} + x_{(j,i)}^{uv} &\le x_e & \hspace{-30pt}\forall (u,v) \in K; \forall e = \{i,j\} \in E \label{eqn:ilp-4}\\
    x_e, x_{(i,j)}^{uv} &\in \{0,1\} \label{eqn:ilp-5}
\end{align} 
Constraint \eqref{eqn:ilp-1} requires that for all $(u,v) \in K$, the sum of the weights of the selected edges corresponding to the pair $(u,v)$ is not more than $t \cdot d_G(u,v)$. Constraints \eqref{eqn:ilp-2}--\eqref{eqn:ilp-3}  
require that the selected edges corresponding to $(u,v) \in K$ form a simple path beginning at $u$ and ending at $v$. Constraint \eqref{eqn:ilp-4} enforces that, if edge $(i,j)$ or $(j,i)$ is selected on some $u$-$v$ path, then its corresponding undirected edge $e$ is selected in the spanner; further, $(i,j)$ and $(j,i)$ cannot both be selected for some pair $(u,v)$. Finally, \eqref{eqn:ilp-5} enforces that all variables are binary.

The number of variables is $|E| + 2|E||K|$ and the number of constraints is $O(|E||K|)$, where $|K| = O(|V|^2)$. Note that the variables $x_{(i,j)}^{uv}$ can be relaxed to be continuous in $[0,1]$.




\subsection{ILP formulation for the MLGS problem} \label{subsection:mlgsilp}
Recall that the MLGS problem generalizes the subsetwise spanner problem, which is a special case of the pairwise spanner problem for $K = S \times S$. Again, we use unordered pairs, i.e., $|K| = \binom{|S|}{2}$.

We generalize the ILP formulation in \eqref{eqn:ilp-obj}--\eqref{eqn:ilp-5} to the MLGS problem as follows. Recall that we can encode the levels in terms of required grades of service $R:V \to \{0,1,\ldots,\ell\}$. Instead of 0--1 indicators $x_e$, let $y_e$ denote the grade of edge $e$ in the multi-level spanner; that is, $y_e = i$ if $e$ appears on level $i$ but not level $i+1$, and $y_e = 0$ if $e$ is absent. The only difference is that for the MLGS problem, we assign grades of service to all $u$-$v$ paths by assigning grades to edges along each $u$-$v$ path. That is, for all $u,v \in T_1$ with $u < v$, the \emph{selected} path from $u$ to $v$ has grade $\min(R(u), R(v))$, which we denote by $m_{uv}$. Note that we only need to require the existence of a path for terminals $u,v \in T_1$, where $u < v$. An ILP formulation for the MLGS problem is as follows.

\begin{align}
    \allowdisplaybreaks
    \text{Minimize} \sum_{e \in E} c_ey_e \text{ subject to} \label{eqn:mlgs-ilp-obj}\\
    \sum_{(i,j) \in E'} x_{(i,j)}^{uv} c_e &\le t \cdot d_G(u,v) & \forall u,v \in T_1; e = \{i,j\} \label{eqn:mlgs-ilp-1}\\
    \sum_{(i,j) \in Out(i)} x_{(i,j)}^{uv} - \sum_{(j,i) \in In(i)} x_{(j,i)}^{uv} &= \begin{cases}
        1 & i = u \\
        -1 & i = v \\
        0 & \text{else}
    \end{cases} & \forall u,v \in T_1; \forall i \in V \label{eqn:mlgs-ilp-2}\\
    \sum_{(i,j) \in Out(i)} x_{(i,j)}^{uv} &\le 1 &\forall u,v \in T_1; \forall i \in V \label{eqn:mlgs-ilp-3}\\
    y_e &\ge m_{uv}x_{(i,j)}^{uv} & \forall u,v \in T_1; \forall \, e = \{i,j\} \label{eqn:mlgs-ilp-4}\\
    y_e &\ge m_{uv}x_{(j,i)}^{uv} & \forall u,v \in T_1; \forall \, e = \{i,j\} \label{eqn:mlgs-ilp-5}\\
    x_{(i,j)}^{uv} &\in \{0,1\}  \label{eqn:mlgs-ilp-6}
\end{align}


Constraints \eqref{eqn:mlgs-ilp-4}--\eqref{eqn:mlgs-ilp-5} enforce that for each pair $u,v \in V$ such that $u < v$, the edges along the selected $u$-$v$ path (not necessarily every $u$-$v$ path) have a grade of service greater than or equal to the minimum grade of service needed to connect $u$ and $v$, that is, $m_{uv}$. If multiple pairs $(u_1, v_1)$, $(u_2, v_2)$, \ldots, $(u_k, v_k)$ use the same edge $e = \{i,j\}$ (possibly in opposite directions), then the grade of edge $e$ should be $y_e = \max(m_{u_1v_1}, m_{u_2v_2}, \ldots, m_{u_kv_k})$. It is implied by \eqref{eqn:mlgs-ilp-4}--\eqref{eqn:mlgs-ilp-5} that $0 \le y_e \le \ell$ in an optimal solution.

\begin{theorem}\label{THM:tSILP}
An optimal solution to the ILP given in \eqref{eqn:ilp-obj}--\eqref{eqn:ilp-5} yields an optimal pairwise spanner of $G$ over a set $K \subset V \times V$.
\end{theorem}

\begin{theorem}\label{THM:MLGSILP}
An optimal solution to the ILP given in \eqref{eqn:mlgs-ilp-obj}--\eqref{eqn:mlgs-ilp-6} yields an optimal solution to the MLGS problem.
\end{theorem}
We give the proofs in Appendices \ref{apdx:tSILP} and \ref{apdx:MLGSILP}.

\subsection{Size reduction techniques}
\label{size_reduction_techniques}
We can reduce the size of the ILP using the following shortest path tests, which works well in practice and also applies to the MLGS problem. Note that we are concerned with the total cost of a solution, not the number of edges.

If $d_G(i,j) < c(i,j)$, for some edge $\{i,j\} \in E$, then we can remove $\{i,j\}$ from the graph, as no min-weight spanner of $G$ uses edge $\{i,j\}$. If $H^*$ is a min-cost pairwise spanner that uses edge $\{i,j\}$, then we can replace $\{i,j\}$ with a shorter $i$-$j$ path $p_{ij}$ without violating the $t$--spanner requirement. In particular, if some $u$-$v$ path uses both edge $\{i,j\}$ as well as some edge(s) along $p_{ij}$, then this path can be rerouted to use only edges in $p_{ij}$ with smaller cost.

We reduce the number of variables needed in the single-level ILP formulation (\eqref{eqn:ilp-obj}--\eqref{eqn:ilp-5}) with the following test: given $u,v \in K$ with $u < v$ and some directed edge $(i,j) \in E'$, if $d_G(u,i) + c(i,j) + d_G(j,v) > t \cdot d_G(u,v)$, then $(i,j)$ cannot possibly be included in the selected $u$-$v$ path, so set $x_{(i,j)}^{uv}=0$. If $(i,j)$ or $(j,i)$ cannot be selected on any $u$-$v$ path, we can safely remove $\{i,j\}$ from $E$.

Conversely, given some directed edge $(i,j) \in E'$, let $G'$ be the directed graph obtained by removing $(i,j)$ from $E'$ (so that $G'$ has $2|E|-1$ edges). For each $u,v \in K$ with $u < v$, if $d_{G'}(u,v) > t \cdot d_G(u,v)$, then edge $(i,j)$ must be in any $u$-$v$ spanner path, so set $x_{(i,j)}^{uv} = 1$. For its corresponding undirected edge $e$, $x_e = 1$.

\section{Experimental results}
\subsection{Setup}


We use the Erd\H{o}s--R\'{e}nyi~\cite{erdos1959random} and Watts--Strogatz \cite{watts1998collective} models to generate random graphs. 
Given a number of vertices, $n$, and probability $p$, the model $\textsc{ER}(n,p)$ assigns an edge to any given pair of vertices with probability $p$. 
An instance of $\textsc{ER}(n,p)$ with $p=(1+\varepsilon)\frac{\ln n}{n}$ is connected with high probability for $\varepsilon>0$~\cite{erdos1959random}).  For our experiments we use $n \in \{20, 40, 60, 80, 100\}$, and $\varepsilon = 1$.

In the Watts-Strogatz model, $\textsc{WS}(n,K,\beta)$, initially we create a ring lattice of constant degree $K$, and then rewire each edge with probability $0\leq \beta \leq 1$ while avoiding self-loops and duplicate edges. The Watts-Strogatz model generates small-world graphs with high clustering coefficients~\cite{watts1998collective}. For our experiments we use $n \in \{20, 40, 60, 80, 100\}$, $K=6$, and $\beta=0.2$. 

An instance of the MLGS problem is characterized by four parameters: graph generator, number of vertices $|V|$, number of levels $\ell$, and stretch factor $t$.  As there is randomness involved, we generated 3 instances for every choice of parameters (e.g., ER, $|V| = 80$, $\ell=3$, $t=2$).

We generated MLGS instances with 1, 2, or 3 levels ($\ell \in \{1, 2, 3\}$), where terminals are selected on each level by randomly sampling $\lfloor {|V| \cdot (\ell-i+1)}/(\ell+1) \rfloor$ vertices on level $i$ so that the size of the terminal sets decreases linearly. As the terminal sets are nested, $T_i$ can be selected by sampling from $T_{i-1}$ (or from $V$ if $i=1$).
We used four different stretch factors in our experiments, $t \in \{1.2, 1.4, 2, 4\}$. Edge weights are randomly selected from $\{1,2,3,\ldots,10\}$. 

\subsubsection{Algorithms and outputs}

We implemented the bottom-up (BU) and top-down (TD) approaches from Section \ref{section:approximation_algorithms} in Python 3.5, as well as the combined approach that selects the better of the two (Section \ref{subsection:other-heuristics}). To evaluate the approximation algorithms and the heuristics, we implemented the ILPs described in Section~\ref{section:ilp} using CPLEX 12.6.2. 
We used the same high-performance computer for all experiments (Lenovo NeXtScale nx360 M5 system with 400 nodes). 

For each instance of the MLGS problem, we compute the costs of the MLGS returned using the bottom-up (BU), the top-down
 (TD), and the combined (min(BU, TD)) approaches, as well as the minimum cost MLGS using the ILP in Section \ref{subsection:mlgsilp}. 
The three heuristics involve a (single-level) subroutine; we used both the heuristic described in Section \ref{heuristic_subsetwise_spanners}, as well as the flow formulation described in Section~\ref{section:ilp} which computes subsetwise spanners to optimality. We compare the algorithms with and without the oracle to assess whether computing (single-level) spanners to optimality significantly improves the overall quality of the solution.

We show the performance ratio for each heuristic in the $y$-axis (defined as the heuristic cost divided by \OPT), and how the ratio depends on the input parameters (number of vertices $|V|$, number of levels $\ell$, and stretch factors $t$).
Finally, we discuss the running time of the ILP. All box plots show the minimum, interquartile range and maximum, aggregated over all instances using the parameter being compared.

\subsection{Results}
\label{section:results}

We first discuss the results for Erd\H{o}s--R{\'e}nyi graphs. 
Figures \ref{oracle_NVR}--\ref{oracle_box} show the results of the oracle top-down, bottom-up, and combined approaches. We show the impact of different parameters (number of vertices $|V|$, number of levels $\ell$, and stretch factors $t$) using line plots for three heuristics separately in Figures \ref{oracle_NVR}-\ref{oracle_SVR}. Figure~\ref{oracle_box} shows the performance of the three heuristics together in box plots. In Figure~\ref{oracle_NVR} we can see that the bottom-up heuristic performs slightly worse for increasing $|V|$, while the top-down heuristic performs slightly better. In Figure~\ref{oracle_LVR} we see that the heuristics perform worse when $\ell$ increases, consistent with the ratios discussed in Section~\ref{section:approximation_algorithms}. In Figure~\ref{oracle_SVR} we show the performance of the heuristics with respect to the stretch factor $t$. In general, the performance becomes worse as $t$ increases.


\begin{figure}[htp]
    \centering
    \begin{subfigure}[b]{0.28\textwidth}
        \includegraphics[width=\textwidth]{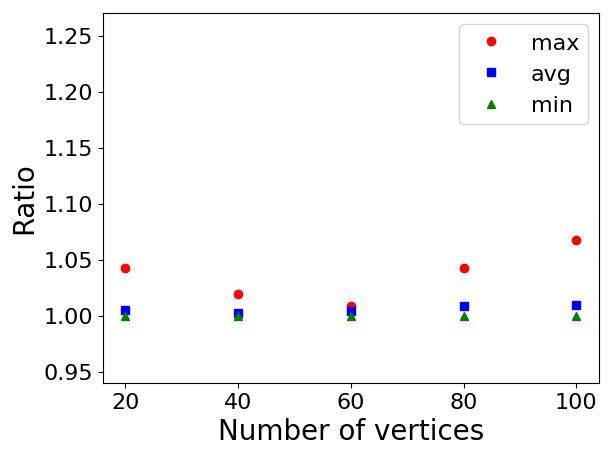}
        \caption{Bottom up}
    \end{subfigure}
    ~
    \begin{subfigure}[b]{0.28\textwidth}
        \includegraphics[width=\textwidth]{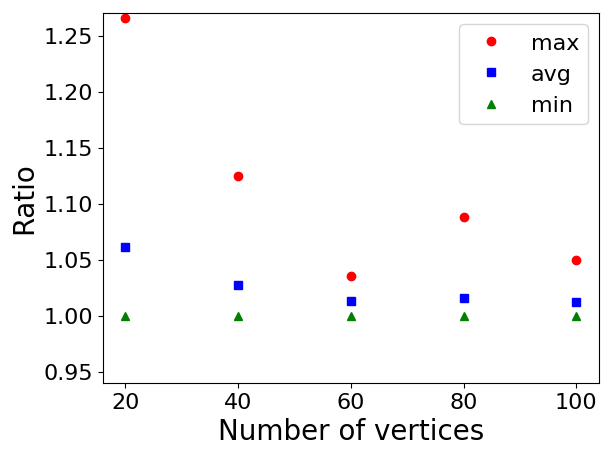}
        \caption{Top down}
    \end{subfigure}
    ~
    \begin{subfigure}[b]{0.28\textwidth}
        \includegraphics[width=\textwidth]{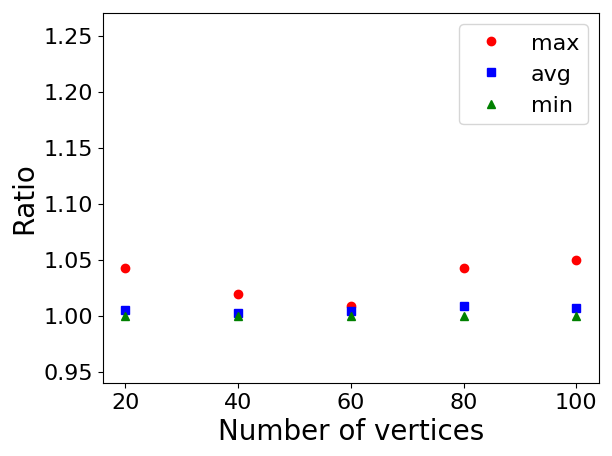}
        \caption{min(BU, TD)}
    \end{subfigure}
    \caption{Performance with oracle on Erd\H{o}s--R{\'e}nyi graphs w.r.t.\ $|V|$. Ratio is defined as the cost of the returned MLGS divided by OPT.} \label{oracle_NVR}
\end{figure}

\begin{figure}[htp]
    \centering
    \begin{subfigure}[b]{0.28\textwidth}
        \includegraphics[width=\textwidth]{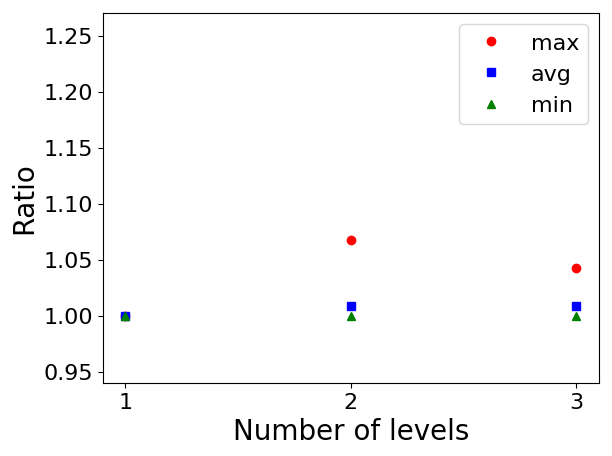}
        \caption{Bottom up}
    \end{subfigure}
    ~
    \begin{subfigure}[b]{0.28\textwidth}
        \includegraphics[width=\textwidth]{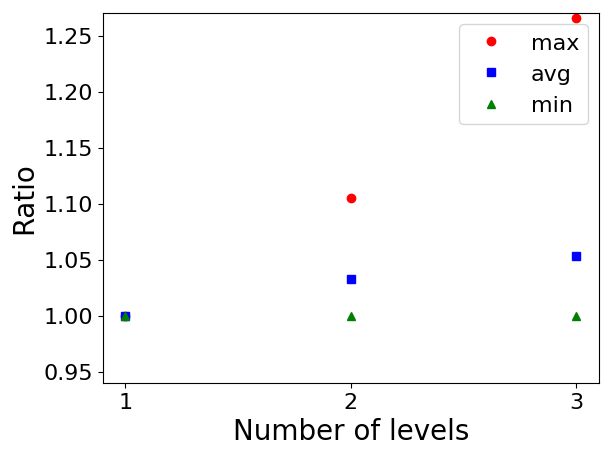}
        \caption{Top down}
    \end{subfigure}
    ~
    \begin{subfigure}[b]{0.28\textwidth}
        \includegraphics[width=\textwidth]{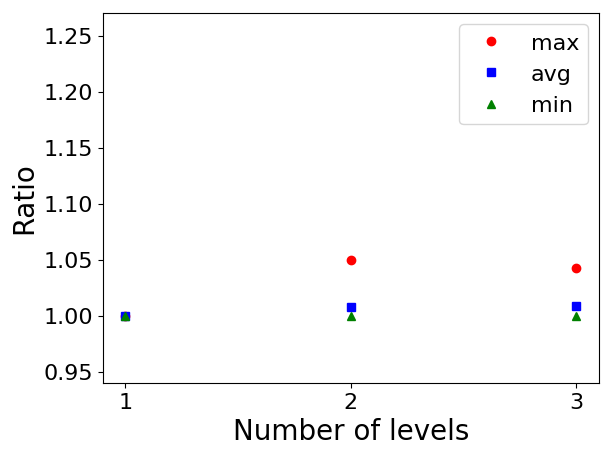}
        \caption{min(BU, TD)}
    \end{subfigure}
    \caption{Performance with oracle on Erd\H{o}s--R{\'e}nyi graphs w.r.t.\ the number of levels} \label{oracle_LVR}
\end{figure}

\begin{figure}[htp]
    \centering
    \begin{subfigure}[b]{0.28\textwidth}
        \includegraphics[width=\textwidth]{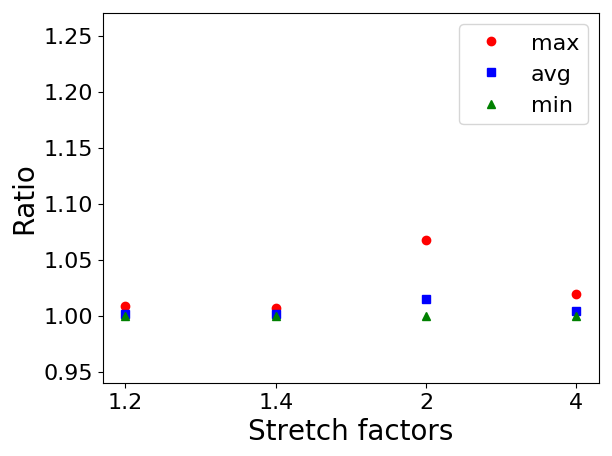}
        \caption{Bottom up}
    \end{subfigure}
    ~
    \begin{subfigure}[b]{0.28\textwidth}
        \includegraphics[width=\textwidth]{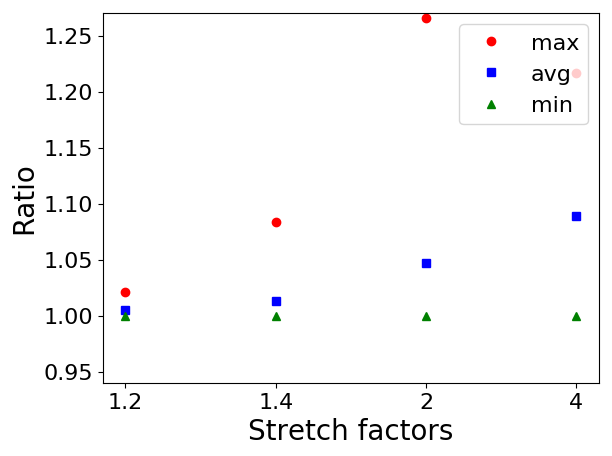}
        \caption{Top down}
    \end{subfigure}
    ~
    \begin{subfigure}[b]{0.28\textwidth}
        \includegraphics[width=\textwidth]{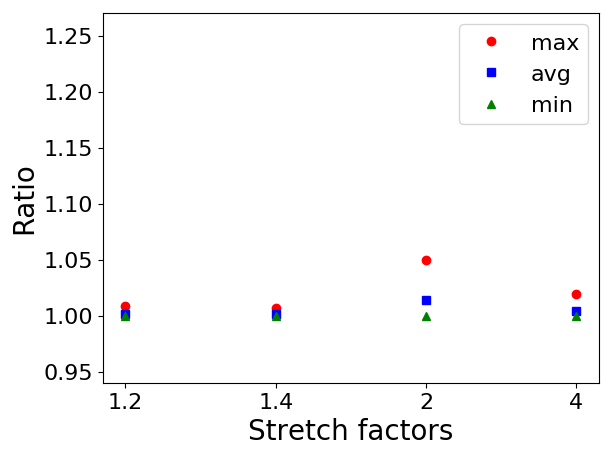}
        \caption{min(BU, TD)}
    \end{subfigure}
    \caption{Performance with oracle on Erd\H{o}s--R{\'e}nyi graphs w.r.t.\ stretch factor} \label{oracle_SVR}
\end{figure}

\begin{figure}[htp]
    \centering
    \begin{subfigure}[b]{0.28\textwidth}
        \includegraphics[width=\textwidth]{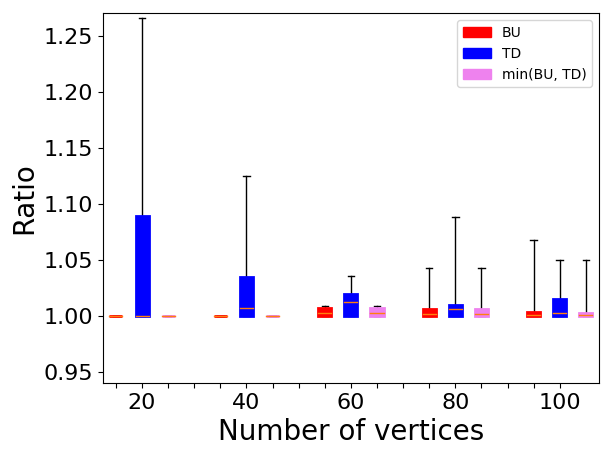}
    \end{subfigure}
    ~
    \begin{subfigure}[b]{0.28\textwidth}
        \includegraphics[width=\textwidth]{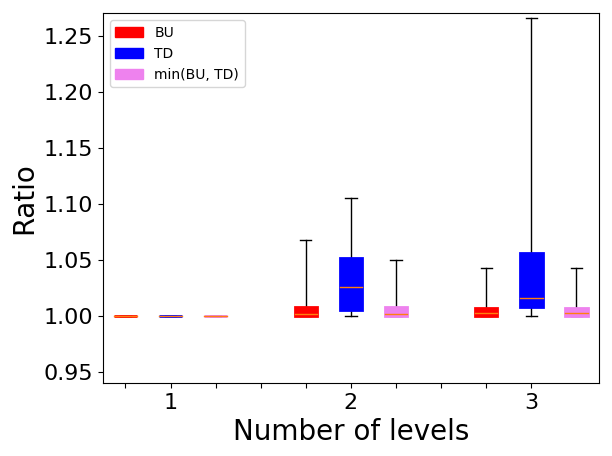}
    \end{subfigure}
    ~
    \begin{subfigure}[b]{0.28\textwidth}
        \includegraphics[width=\textwidth]{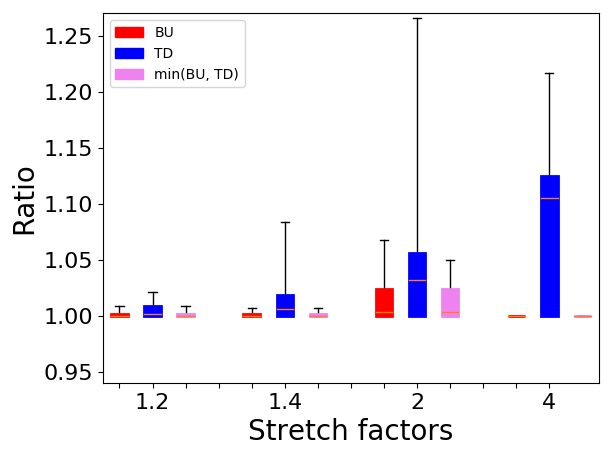}
    \end{subfigure}
    \caption{Performance with oracle on Erd\H{o}s--R{\'e}nyi graphs w.r.t.\ the number of vertices, the number of levels, and the stretch factors} \label{oracle_box}
\end{figure}

The most time consuming part of the experiment is the execution time of the ILP for solving MLGS instances optimally. The running time of the heuristics is significantly smaller compared to that of the ILP. Hence, we first show the running times of the exact solution of the MLGS instances in Figure~\ref{time_box}. We show the running time with respect to the number of vertices $|V|$, number of levels $\ell$, and stretch factors $t$. For all parameters, the running time tends to increase as the size of the parameter increases. 
In particular, the running time with stretch factor 4 (Fig. \ref{time_box}, right) was much worse, as there are many more $t$-spanner paths to consider, and the size reduction techniques in Section~\ref{size_reduction_techniques} are less effective at reducing instance size. We show the running times of for computing oracle bottom-up, top-down and combined solutions in Figure~\ref{time_box_heu}. 

\begin{figure}[htp]
    \centering
    \begin{subfigure}[b]{0.28\textwidth}
        \includegraphics[width=\textwidth]{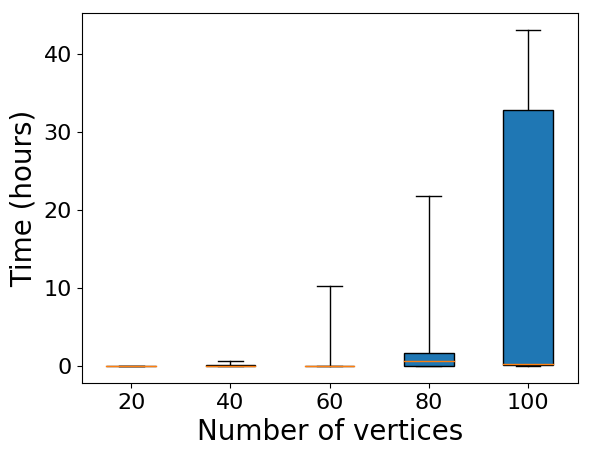}
    \end{subfigure}
    ~
    \begin{subfigure}[b]{0.28\textwidth}
        \includegraphics[width=\textwidth]{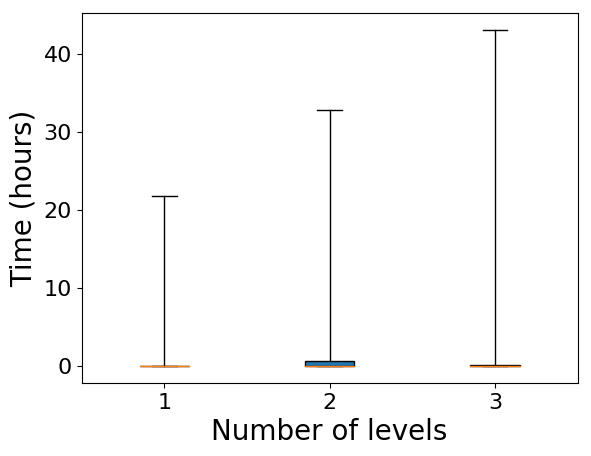}
    \end{subfigure}
    ~
    \begin{subfigure}[b]{0.28\textwidth}
        \includegraphics[width=\textwidth]{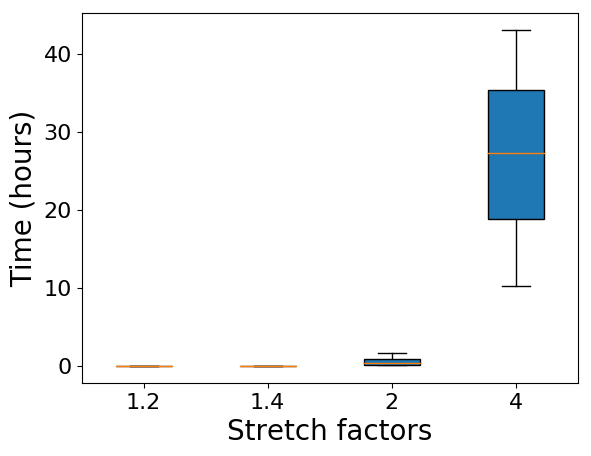}
    \end{subfigure}
    \caption{Experimental running times for computing exact solutions on Erd\H{o}s--R{\'e}nyi graphs w.r.t.\ the number of vertices, the number of levels, and the stretch factors} \label{time_box}
\end{figure}

\begin{figure}[htp]
    \centering
    \begin{subfigure}[b]{0.28\textwidth}
        \includegraphics[width=\textwidth]{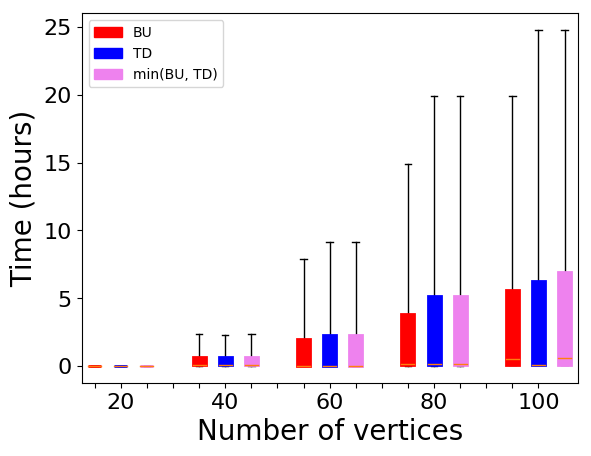}
    \end{subfigure}
    ~
    \begin{subfigure}[b]{0.28\textwidth}
        \includegraphics[width=\textwidth]{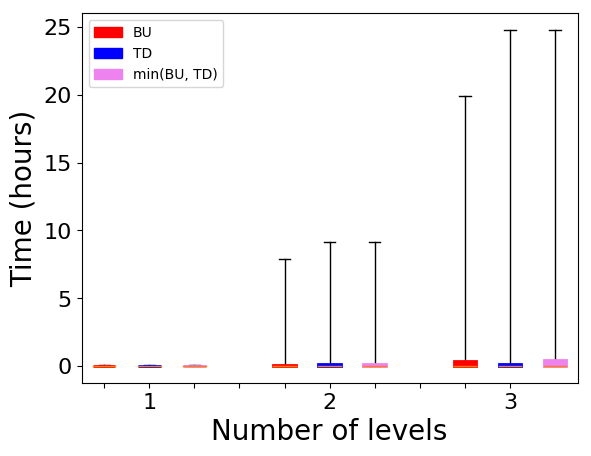}
    \end{subfigure}
    ~
    \begin{subfigure}[b]{0.28\textwidth}
        \includegraphics[width=\textwidth]{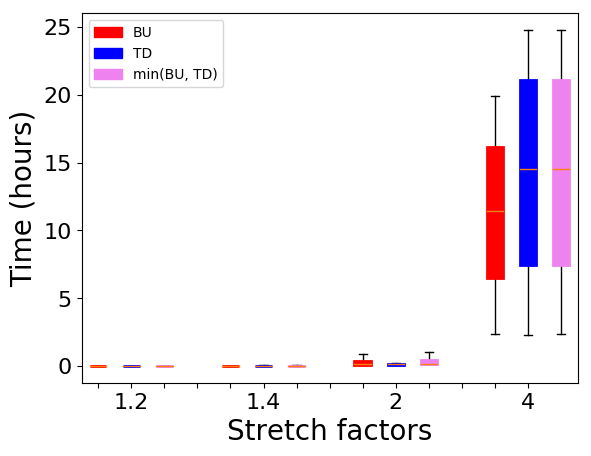}
    \end{subfigure}
    \caption{Experimental running times for computing oracle bottom-up, top-down and combined solutions on Erd\H{o}s--R{\'e}nyi graphs w.r.t.\ the number of vertices, the number of levels, and the stretch factors} \label{time_box_heu}
\end{figure}

The ILP is too computationally expensive for larger input sizes and this is where the heurstic can be particularly useful. We now consider a similar experiment using the heuristic to compute subsetwise spanners, as described in Section \ref{heuristic_subsetwise_spanners}. We show the impact of different parameters (number of vertices $|V|$, number of levels $\ell$, and stretch factors $t$) using scatter plots for three heuristics separately in Figures \ref{approx_NVR}--\ref{approx_SVR}.  Figure~\ref{approx_box} shows the performance of the three heuristics together in box plots. We can see that the heuristics perform very well in practice. Notably when the heuristic is used in place of the ILP (Fig \ref{consistent_approx_box_ER}), the running times decrease for larger stretch factors.

\begin{figure}[htp]
    \centering
    \begin{subfigure}[b]{0.28\textwidth}
        \includegraphics[width=\textwidth]{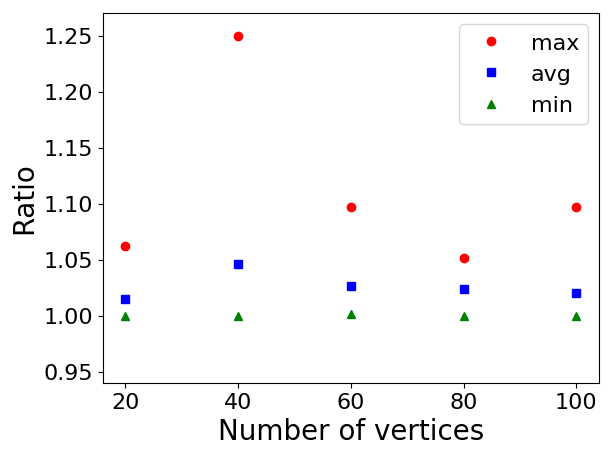}
        \caption{Bottom up}
    \end{subfigure}
    ~
    \begin{subfigure}[b]{0.28\textwidth}
        \includegraphics[width=\textwidth]{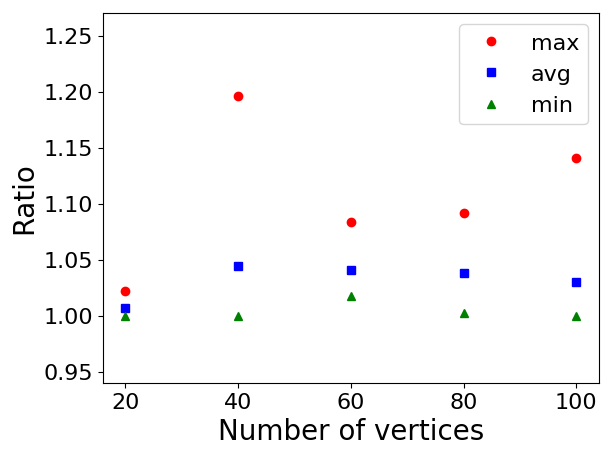}
        \caption{Top down}
    \end{subfigure}
    ~
    \begin{subfigure}[b]{0.28\textwidth}
        \includegraphics[width=\textwidth]{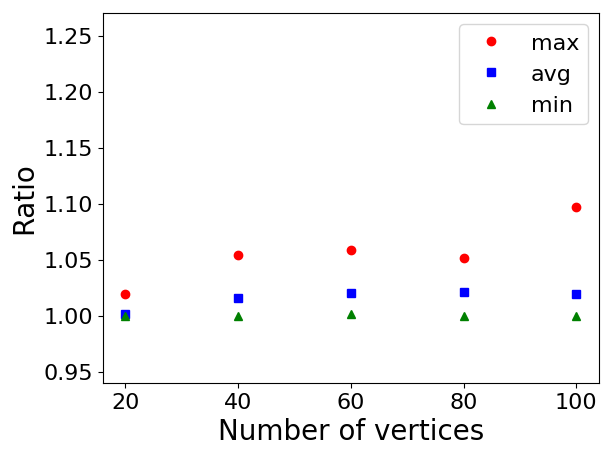}
        \caption{min(BU, TD)}
    \end{subfigure}
    \caption{Performance without oracle on Erd\H{o}s--R{\'e}nyi graphs w.r.t.\ $|V|$} \label{approx_NVR}
\end{figure}

\begin{figure}[htp]
    \centering
    \begin{subfigure}[b]{0.28\textwidth}
        \includegraphics[width=\textwidth]{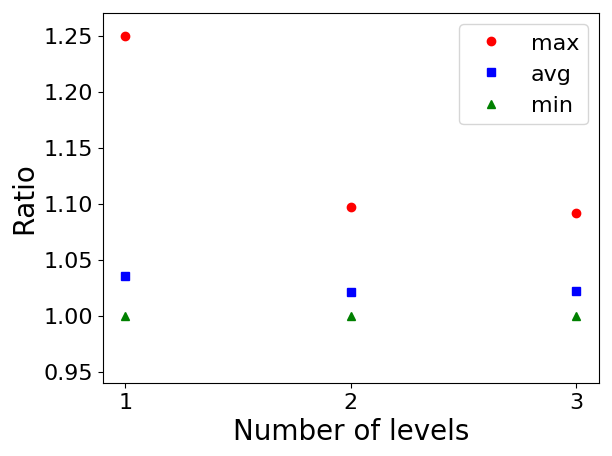}
        \caption{Bottom up}
    \end{subfigure}
    ~
    \begin{subfigure}[b]{0.28\textwidth}
        \includegraphics[width=\textwidth]{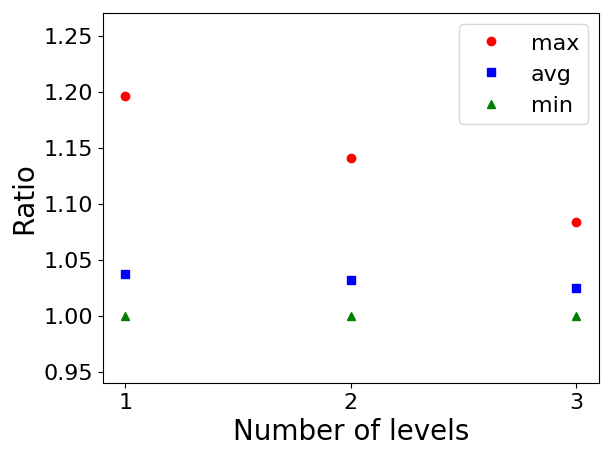}
        \caption{Top down}
    \end{subfigure}
    ~
    \begin{subfigure}[b]{0.28\textwidth}
        \includegraphics[width=\textwidth]{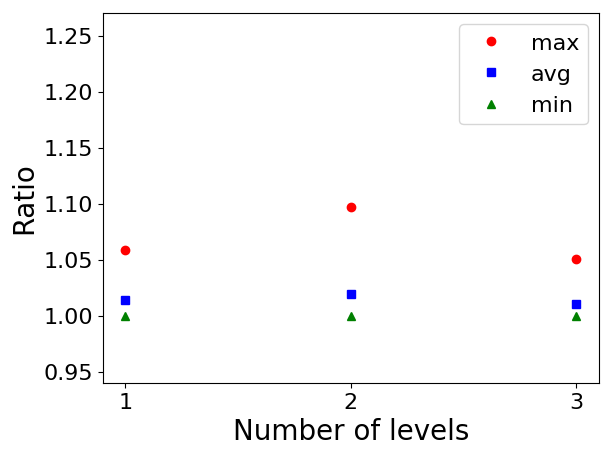}
        \caption{min(BU, TD)}
    \end{subfigure}
    \caption{Performance without oracle on Erd\H{o}s--R{\'e}nyi graphs w.r.t.\ the number of levels} \label{approx_LVR}
\end{figure}

\begin{figure}[htp]
    \centering
    \begin{subfigure}[b]{0.28\textwidth}
        \includegraphics[width=\textwidth]{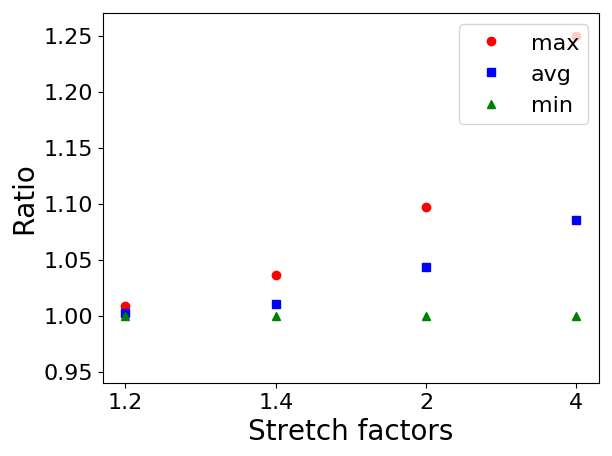}
        \caption{Bottom up}
    \end{subfigure}
    ~
    \begin{subfigure}[b]{0.28\textwidth}
        \includegraphics[width=\textwidth]{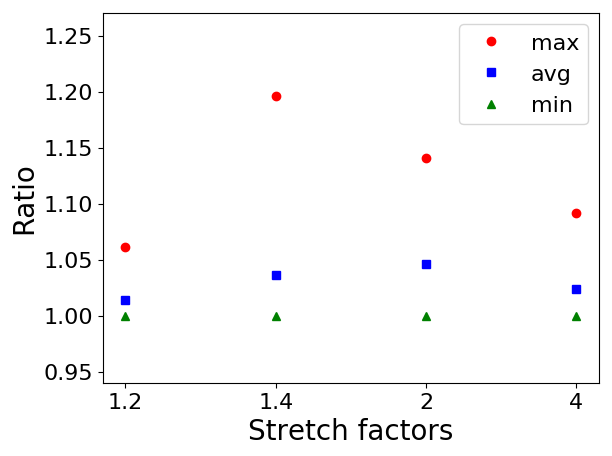}
        \caption{Top down}
    \end{subfigure}
    ~
    \begin{subfigure}[b]{0.28\textwidth}
        \includegraphics[width=\textwidth]{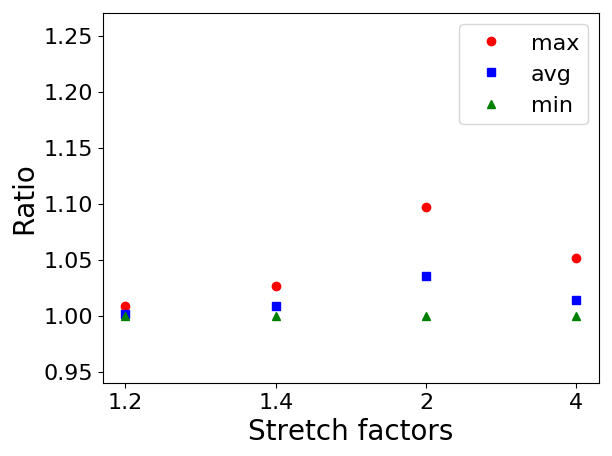}
        \caption{min(BU, TD)}
    \end{subfigure}
    \caption{Performance without oracle on Erd\H{o}s--R{\'e}nyi graphs w.r.t.\ the stretch factors} \label{approx_SVR}
\end{figure}

\begin{figure}[htp]
    \centering
    \begin{subfigure}[b]{0.28\textwidth}
        \includegraphics[width=\textwidth]{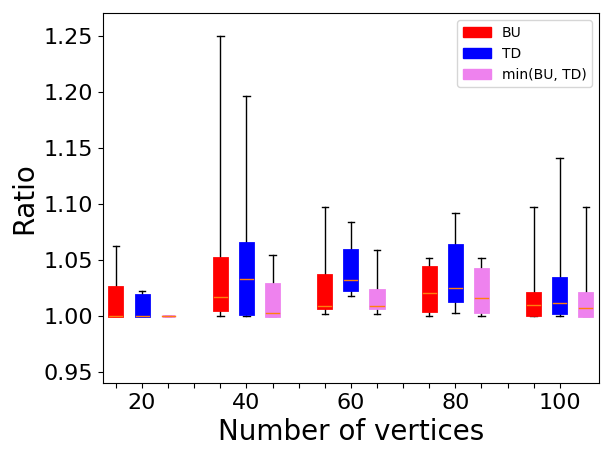}
    \end{subfigure}
    ~
    \begin{subfigure}[b]{0.28\textwidth}
        \includegraphics[width=\textwidth]{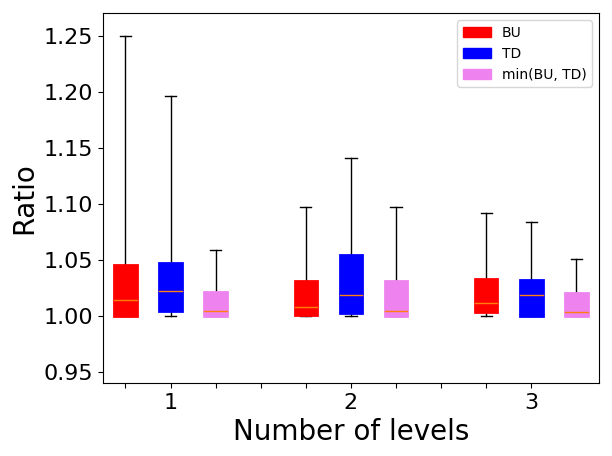}
    \end{subfigure}
    ~
    \begin{subfigure}[b]{0.28\textwidth}
        \includegraphics[width=\textwidth]{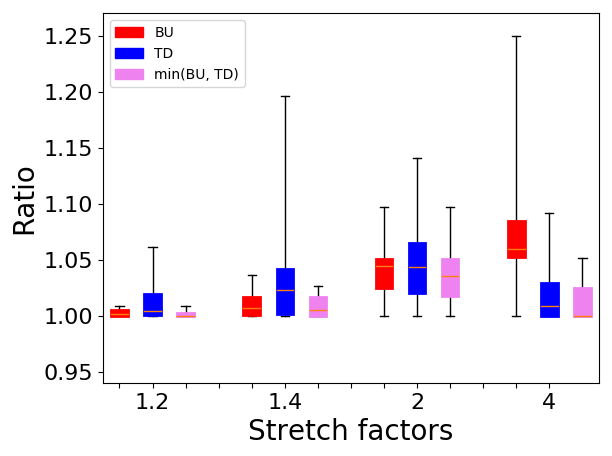}
    \end{subfigure}
    \caption{Performance without oracle on Erd\H{o}s--R{\'e}nyi graphs w.r.t.\ the number of vertices, the number of levels, and the stretch factors} \label{approx_box}
\end{figure}

\begin{figure}[htp]
    \centering
    \begin{subfigure}[b]{0.28\textwidth}
        \includegraphics[width=\textwidth]{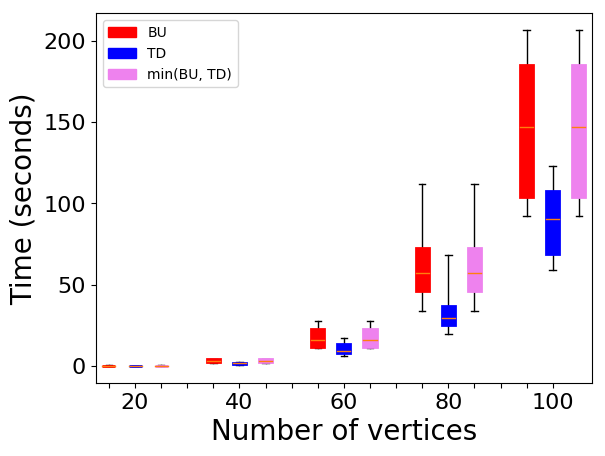}
    \end{subfigure}
    ~
    \begin{subfigure}[b]{0.28\textwidth}
        \includegraphics[width=\textwidth]{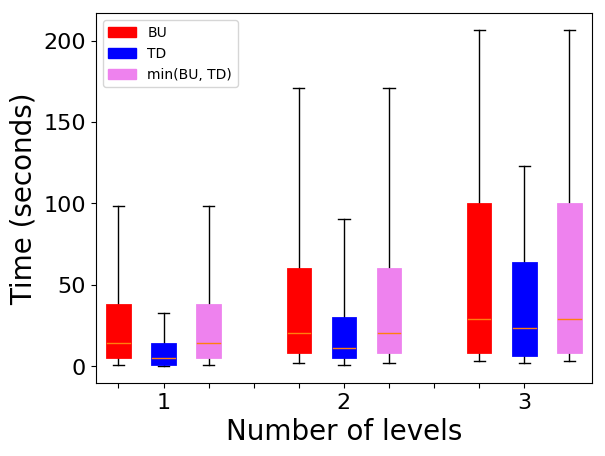}
    \end{subfigure}
    ~
    \begin{subfigure}[b]{0.28\textwidth}
        \includegraphics[width=\textwidth]{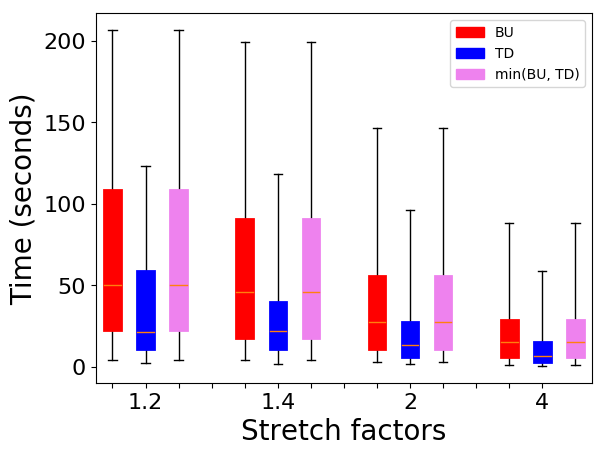}
    \end{subfigure}
    \caption{Experimental running times for computing heuristic bottom-up, top-down and combined solutions on Erd\H{o}s--R{\'e}nyi graphs w.r.t.\ the number of vertices, the number of levels, and the stretch factors} \label{approx_time_box_heu}
\end{figure}


We also analyzed graphs generated from the Watts--Strogatz model and the results are shown in Appendix~\ref{apdx:ExperimentalResults}.

Our final experiments test the heuristic performance on a set of larger graphs. We generated the graphs using the Erd\H{o}s--R{\'e}nyi model, with $|V| \in \{100, 200, 300, 400\}$. We evaluated more levels ($\ell \in \{2, 4, 6, 8, 10\}$) with stretch factors $t \in \{1.2, 1.4, 2, 4\}$. We show the performance of heuristic bottom-up and top-down in Appendix \ref{apdx:ER-large}. Here, the ratio is determined by dividing the BU or TD cost by $\min(BU,TD)$ (as computing the optimal MLGS would be too time-consuming). The results indicate that while running times increase with larger input graphs, the number of levels and the stretch factors seem to have little impact on performance.

\section{Discussion and conclusion}
We introduced a generalization of the subsetwise spanner problem to multiple levels or grades of service. Our proposed ILP formulation requires only a polynomial size of variables and constraints, which is an improvement over the previous formulation given by Sigurd and Zachariasen~\cite{sigrd04}. We also proposed improved formulations which  work well for small values of the stretch factor $t$. It would be worthwhile to consider whether even better ILP formulations can be found for computing graph spanners and their multi-level variants. We showed that both  the approximation algorithms and the heuristics work well in practice on several different types of graphs, with different  number of levels and different stretch factors. 

We only considered a stretch factor $t$  that is the same for all levels in the multi-level spanner, as well as a fairly specific definition of cost. 
It would be interesting to investigate more general multi-level or grade-of-service spanner problems, including ones with varying stretch factors (e.g., in which more important terminals require a smaller or larger stretch factors), different definitions of cost, and spanners with other requirements, such as bounded diameters or degrees.


%
%
%
\bibliographystyle{splncs04}
\bibliography{references}

\appendix
\section{Proof of Theorem \ref{thm:tdbu}} \label{apdx:tdbu}
\begin{proof}
We use the simple algebraic fact that $\min \{x,y\} \le \alpha x + (1-\alpha)y$ for all $x,y \in \mathbb{R}$ and $\alpha \in [0,1]$. Here, we can also use the fact that $\MIN_1 \le \OPT_1 + \OPT_2 + \ldots + \OPT_{\ell}$, as the RHS equals the cost of $G_1^*$, which is some subsetwise $(T_1 \times T_1)$-spanner. Combining, we have
\begin{align*}
    \min(\TOP, \BOT) &\le \alpha \sum_{i=1}^{\ell} \frac{i(i+1)}{2} \OPT_i + (1-\alpha)\ell \sum_{i=1}^{\ell} \OPT_i \\
     &= \sum_{i=1}^{\ell} \left[ \left(\frac{i(i+1)}{2} - \ell\right)\alpha + \ell \right] \rho \OPT_i \\
\end{align*}

Since we are comparing $\min\{\TOP, \BOT\}$ to $r \cdot \OPT$ for some approximation ratio $r > 1$, we can compare coefficients and find the smallest $r \ge 1$ such that the system of inequalities
	\begin{align*}
	\left(\frac{\ell(\ell+1)}{2} - \ell\right) \alpha + \ell \rho &\le \ell r \\
	\left(\frac{(\ell-1)\ell}{2} - \ell\right) \alpha + \ell \rho &\le (\ell-1)r \\[-0.8em]
	&\vdots \\[-0.8em]
	\left(\frac{2 \cdot 1}{2} - \ell\right) \alpha + \ell \rho &\le r
	\end{align*}
	has a solution $\alpha \in [0,1]$. 
	Adding the first inequality to $\ell/2$ times the last inequality yields $\frac{\ell^2 + 2\ell}{2} \le \frac{3\ell r}{2}$, or $r \ge \frac{\ell+2}{3}$.  
	Also, it can be shown algebraically that $(r, \alpha) = (\frac{\ell+2}{3}, \frac{2}{3})$ simultaneously satisfies the above inequalities. 
	This implies that $\min\{\TOP, \BOT\} \le \frac{\ell+2}{3}\rho \cdot \OPT$.\hfill$\square$
\end{proof}
 \section{Proof of Theorem \ref{THM:tSILP}} \label{apdx:tSILP}
\begin{proof}
Let $H^*$ denote an optimal pairwise spanner of $G$ with stretch factor $t$, and let $\OPT$ denote the cost of $H^*$. Let $\OPT_{ILP}$ denote the minimum cost of the objective in the ILP \eqref{eqn:ilp-obj}. First, given a minimum cost $t$--spanner $H^*(V,E^*)$, a solution to the ILP can be constructed as follows:
for each edge $e \in E^*$, set $x_e = 1$. Then for each unordered pair $(u,v) \in K$ with $u < v$, compute a shortest path $p_{uv}$ from $u$ to $v$ in $H^*$, and set $x_{(i,j)}^{uv} = 1$ for each edge along this path, and $x_{(i,j)}^{uv} = 0$ if $(i,j)$ is not on $p_{uv}$.

As each shortest path $p_{uv}$ necessarily has cost $\le t \cdot d_G(u,v)$, constraint (\ref{eqn:ilp-1}) is satisfied. Constraints \eqref{eqn:ilp-2}--\eqref{eqn:ilp-3} are satisfied as $p_{uv}$ is a simple $u$-$v$ path. Constraint \eqref{eqn:ilp-4} also holds, as $p_{uv}$ should not traverse the same edge twice in opposite directions. In particular, every edge in $H^*$ appears on some shortest path; otherwise, removing such an edge yields a pairwise spanner of lower cost. Hence $\OPT_{ILP} \le \OPT$.

Conversely, an optimal solution to the ILP induces a feasible $t$--spanner $H$. Consider an unordered pair $(u,v) \in K$ with $u < v$, and the set of decision variables satisfying $x_{(i,j)}^{uv} = 1$. By \eqref{eqn:ilp-2} and \eqref{eqn:ilp-3}, these chosen edges form a simple path from $u$ to $v$. The sum of the weights of these edges is at most $t \cdot d_G(u,v)$ by \eqref{eqn:ilp-1}. Then by constraint \eqref{eqn:ilp-4}, the chosen edges corresponding to $(u,v)$ appear in the spanner, which is induced by the set of edges $e$ with $x_e = 1$. Hence $\OPT \le \OPT_{ILP}$.

Combining the above observations, we see that $\OPT=\OPT_{ILP}$.\hfill$\square$
\end{proof}

\section{Proof of Theorem \ref{THM:MLGSILP}} \label{apdx:MLGSILP}
\begin{proof}
Given an optimal solution to the ILP with cost $\OPT_{ILP}$, construct an MLGS by letting $G_i = (V, E_i)$ where $E_i = \{e \in E \mid y_e \ge i\}$. This clearly gives a nested sequence of subgraphs. Let $u$ and $v$ be terminals in $T_i$ (not necessarily of required grade $R(\cdot) = i$), with $u < v$, and consider the set of all variables of the form $x_{(i,j)}^{uv}$ equal to 1. By \eqref{eqn:mlgs-ilp-1}--\eqref{eqn:mlgs-ilp-3}, these selected edges form a path from $u$ to $v$ of length at most $t \cdot d_G(u,v)$, while constraints \eqref{eqn:mlgs-ilp-4}--\eqref{eqn:mlgs-ilp-5} imply that these selected edges have grade at least $m_{uv} \ge i$, so the selected path is contained in $E_i$. Hence $G_i$ is a subsetwise $(T_i \times T_i)$--spanner for $G$ with stretch factor $t$, and the optimal ILP solution gives a feasible MLGS.

Given an optimal MLGS with cost $\OPT$, we can construct a feasible ILP solution with the same cost in a way similar to the proof of Theorem \ref{THM:tSILP}. For each $u,v \in T_1$ with $u < v$, set $m_{uv} = \min(R(u), R(v))$. Compute a shortest path in $G_{m_{uv}}$ from $u$ to $v$, and set $x_{(i,j)}^{uv} = 1$ for all edges along this path. Then for each $e \in E$, consider all pairs $(u_1, v_1), \ldots, (u_k,v_k)$ that use either $(i,j)$ or $(j,i)$, and set $y_e = \max(m_{u_1v_1}, m_{u_2v_2}, \ldots, m_{u_kv_k})$. In particular, $y_e$ is not larger than the grade of $e$ in the MLGS, otherwise this would imply $e$ is on some $u$-$v$ path at grade greater than its grade of service in the actual solution.\hfill$\square$
\end{proof}

\section{Experimental results on graphs generated using Watts-Strogatz} \label{apdx:ExperimentalResults}

The results for graphs generated from the Watts--Strogatz model are shown in Figures \ref{oracle_NVR_WS}--\ref{approx_box_WS}, which are organized in the same way as for Erd\H{o}s--R\'enyi. 

\begin{figure}[htp]
    \centering
    \begin{subfigure}[b]{0.28\textwidth}
        \includegraphics[width=\textwidth]{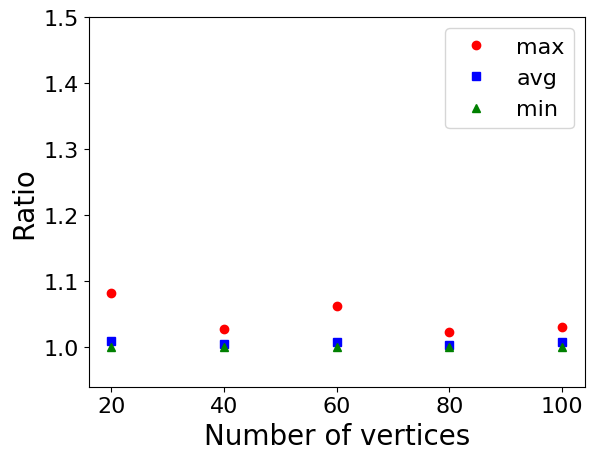}
        \caption{Bottom up}
    \end{subfigure}
    ~
    \begin{subfigure}[b]{0.28\textwidth}
        \includegraphics[width=\textwidth]{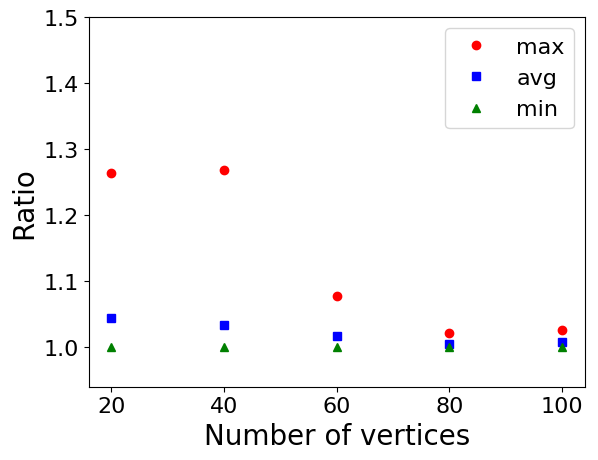}
        \caption{Top down}
    \end{subfigure}
    ~
    \begin{subfigure}[b]{0.28\textwidth}
        \includegraphics[width=\textwidth]{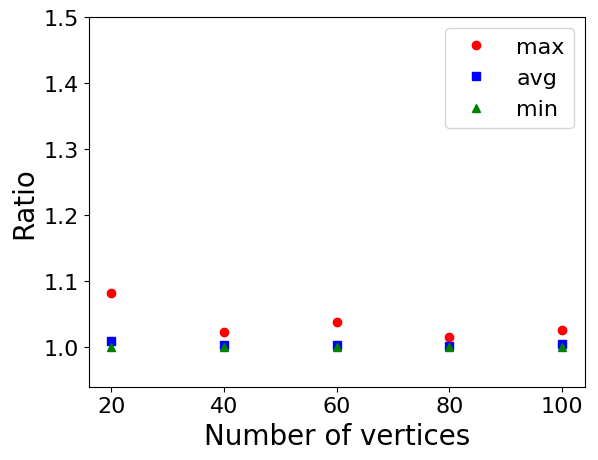}
        \caption{min(BU, TD)}
    \end{subfigure}
    \caption{Performance with oracle on Watts--Strogatz graphs w.r.t.\ the number of vertices} \label{oracle_NVR_WS}
\end{figure}

\begin{figure}[htp]
    \centering
    \begin{subfigure}[b]{0.28\textwidth}
        \includegraphics[width=\textwidth]{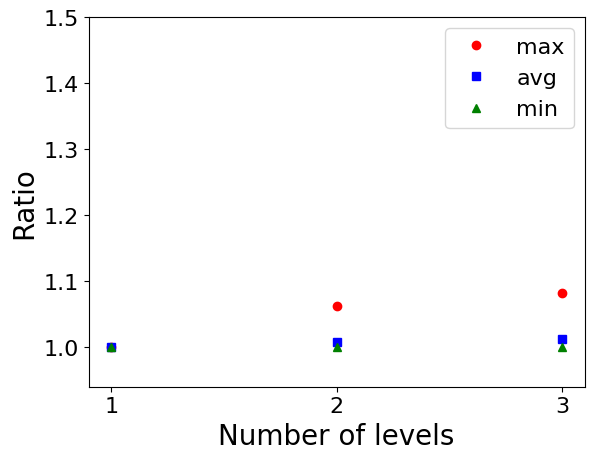}
        \caption{Bottom up}
    \end{subfigure}
    ~
    \begin{subfigure}[b]{0.28\textwidth}
        \includegraphics[width=\textwidth]{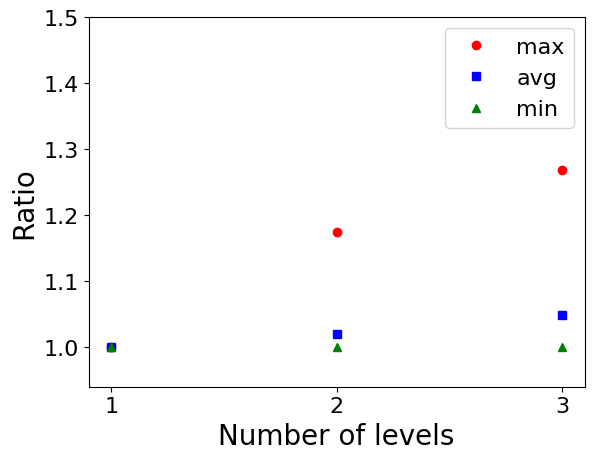}
        \caption{Top down}
    \end{subfigure}
    ~
    \begin{subfigure}[b]{0.28\textwidth}
        \includegraphics[width=\textwidth]{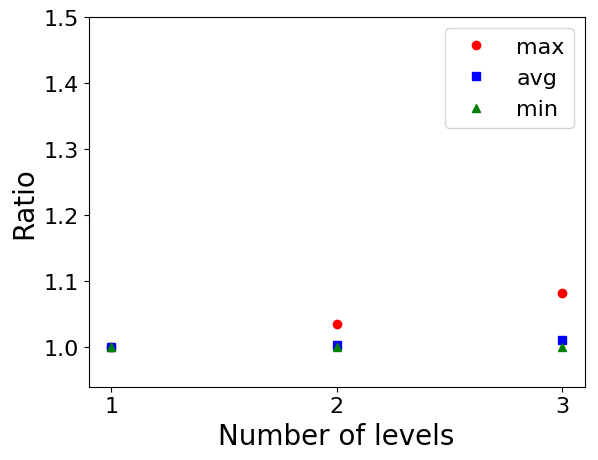}
        \caption{min(BU, TD)}
    \end{subfigure}
    \caption{Performance with oracle on Watts--Strogatz graphs w.r.t.\ the number of levels} \label{oracle_LVR_WS}
\end{figure}

\begin{figure}[htp]
    \centering
    \begin{subfigure}[b]{0.28\textwidth}
        \includegraphics[width=\textwidth]{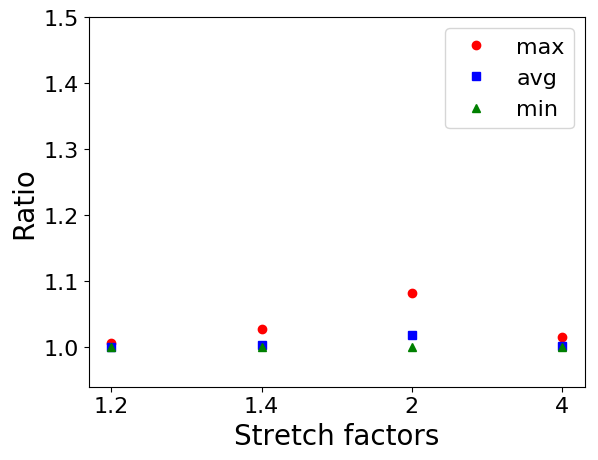}
        \caption{Bottom up}
    \end{subfigure}
    ~
    \begin{subfigure}[b]{0.28\textwidth}
        \includegraphics[width=\textwidth]{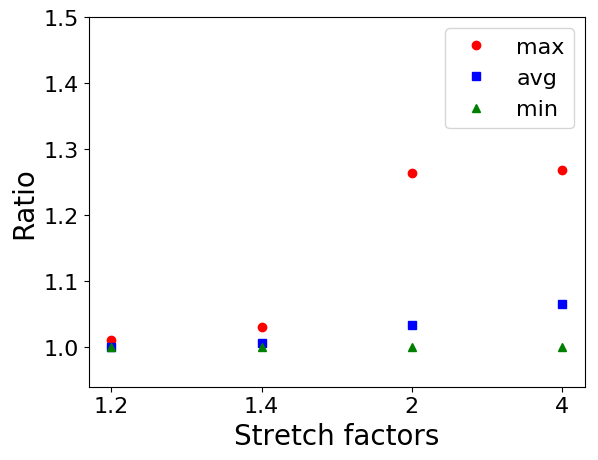}
        \caption{Top down}
    \end{subfigure}
    ~
    \begin{subfigure}[b]{0.28\textwidth}
        \includegraphics[width=\textwidth]{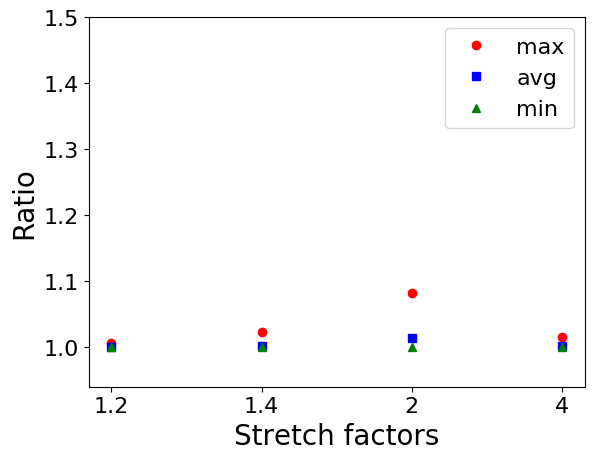}
        \caption{min(BU, TD)}
    \end{subfigure}
    \caption{Performance with oracle on Watts--Strogatz graphs w.r.t.\ the stretch factors} \label{oracle_SVR_WS}
\end{figure}

\begin{figure}[htp]
    \centering
    \begin{subfigure}[b]{0.28\textwidth}
        \includegraphics[width=\textwidth]{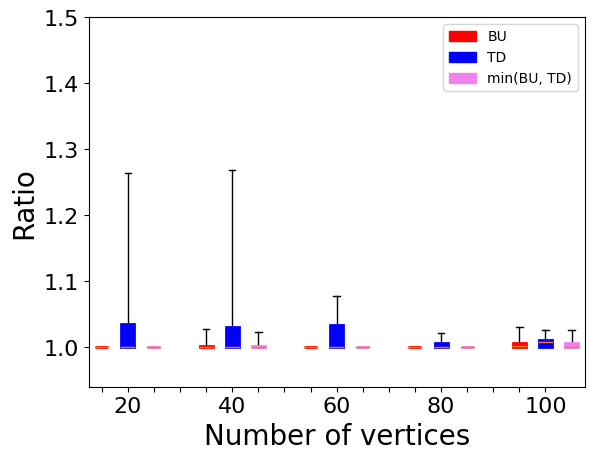}
    \end{subfigure}
    ~
    \begin{subfigure}[b]{0.28\textwidth}
        \includegraphics[width=\textwidth]{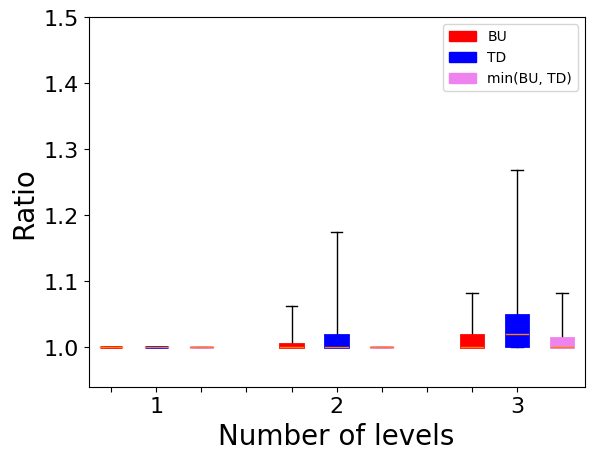}
    \end{subfigure}
    ~
    \begin{subfigure}[b]{0.28\textwidth}
        \includegraphics[width=\textwidth]{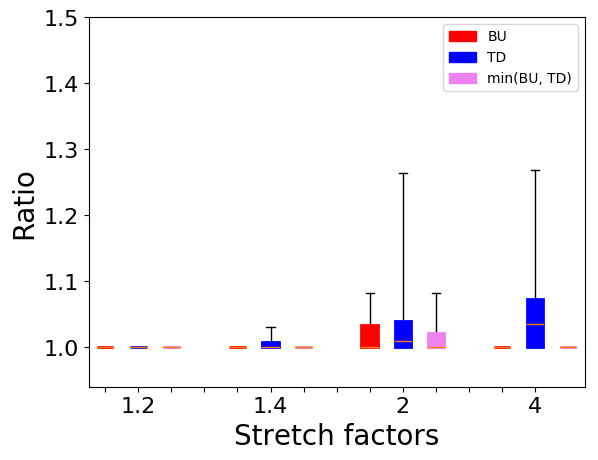}
    \end{subfigure}
    \caption{Performance with oracle on Watts--Strogatz graphs w.r.t.\ the number of vertices, the number of levels, and the stretch factors} \label{oracle_box_WS}
\end{figure}

\begin{figure}[htp]
    \centering
    \begin{subfigure}[b]{0.28\textwidth}
        \includegraphics[width=\textwidth]{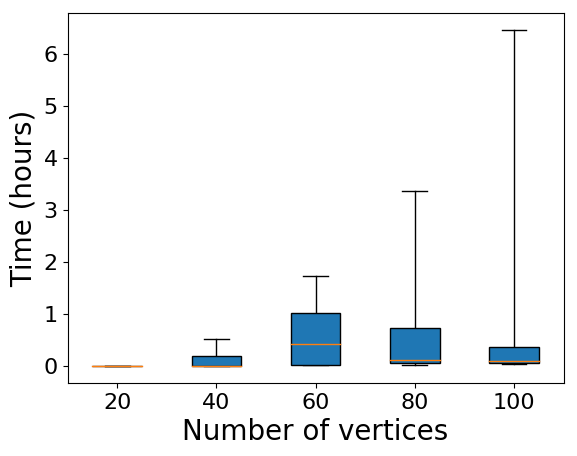}
    \end{subfigure}
    ~
    \begin{subfigure}[b]{0.28\textwidth}
        \includegraphics[width=\textwidth]{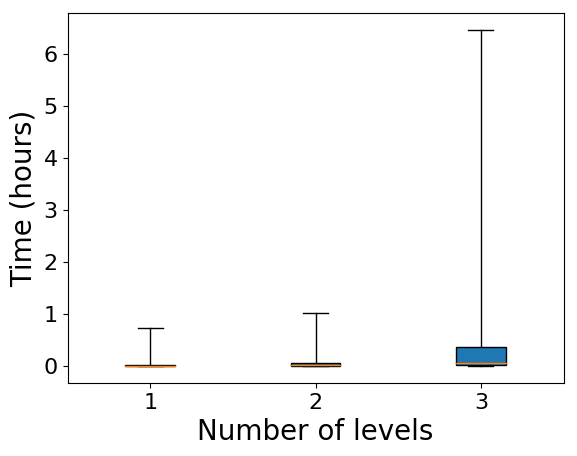}
    \end{subfigure}
    ~
    \begin{subfigure}[b]{0.28\textwidth}
        \includegraphics[width=\textwidth]{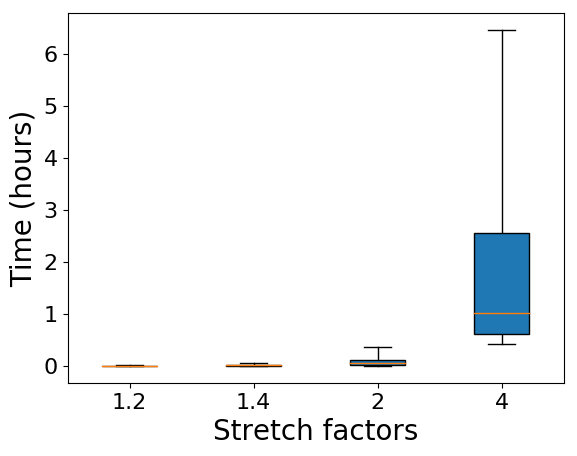}
    \end{subfigure}
    \caption{Experimental running times for computing exact solutions on Watts--Strogatz graphs w.r.t.\ the number of vertices, the number of levels, and the stretch factors} \label{time_box_WS}
\end{figure}

\begin{figure}[htp]
    \centering
    \begin{subfigure}[b]{0.28\textwidth}
        \includegraphics[width=\textwidth]{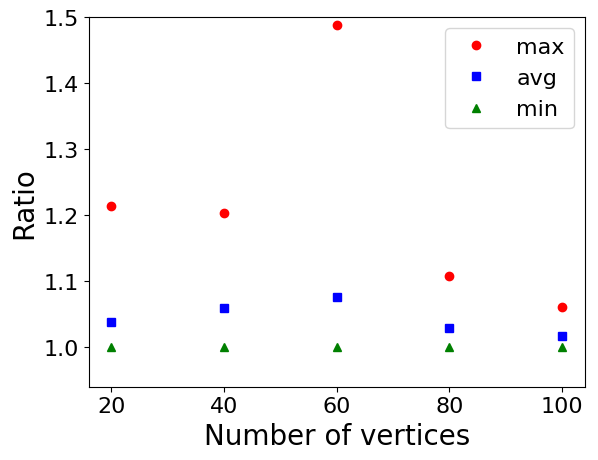}
        \caption{Bottom up}
    \end{subfigure}
    ~
    \begin{subfigure}[b]{0.28\textwidth}
        \includegraphics[width=\textwidth]{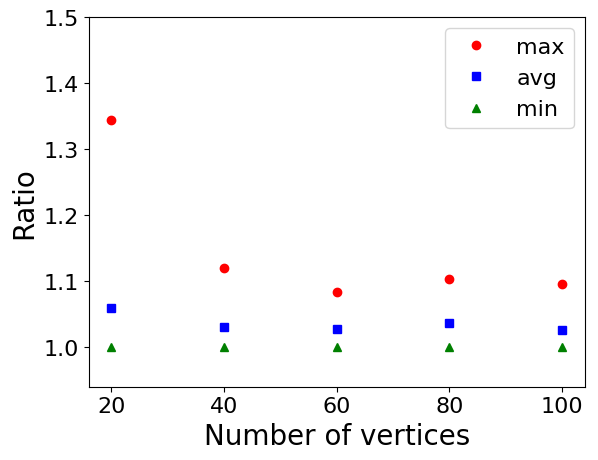}
        \caption{Top down}
    \end{subfigure}
    ~
    \begin{subfigure}[b]{0.28\textwidth}
        \includegraphics[width=\textwidth]{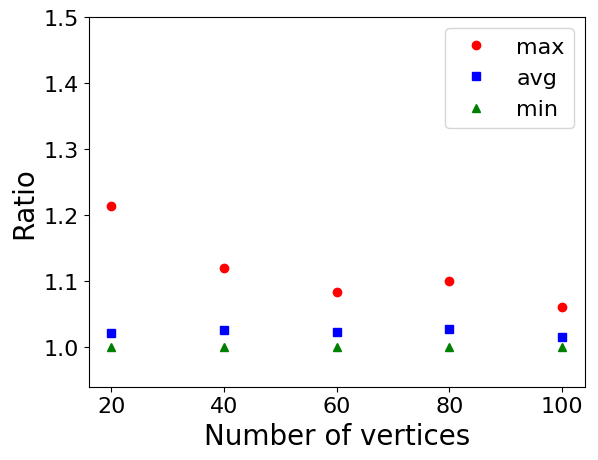}
        \caption{min(BU, TD)}
    \end{subfigure}
    \caption{Performance without oracle on Watts--Strogatz graphs w.r.t.\ the number of vertices} \label{approx_NVR_WS}
\end{figure}

\begin{figure}[htp]
    \centering
    \begin{subfigure}[b]{0.28\textwidth}
        \includegraphics[width=\textwidth]{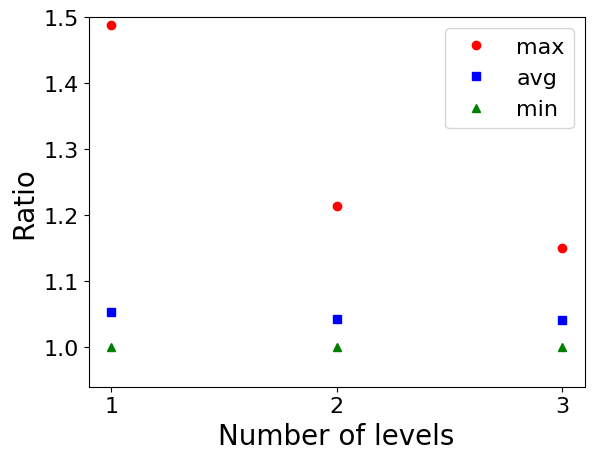}
        \caption{Bottom up}
    \end{subfigure}
    ~
    \begin{subfigure}[b]{0.28\textwidth}
        \includegraphics[width=\textwidth]{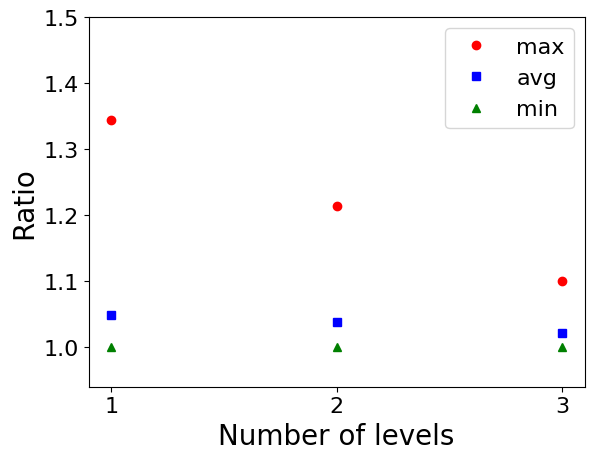}
        \caption{Top down}
    \end{subfigure}
    ~
    \begin{subfigure}[b]{0.28\textwidth}
        \includegraphics[width=\textwidth]{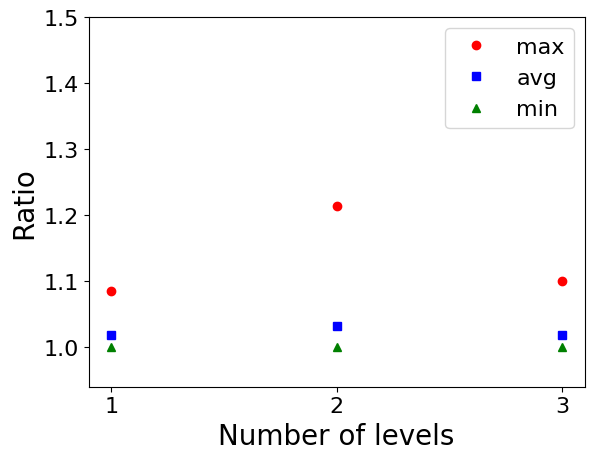}
        \caption{min(BU, TD)}
    \end{subfigure}
    \caption{Performance without oracle on Watts--Strogatz graphs w.r.t.\ the number of levels} \label{approx_LVR_WS}
\end{figure}

\begin{figure}[htp]
    \centering
    \begin{subfigure}[b]{0.28\textwidth}
        \includegraphics[width=\textwidth]{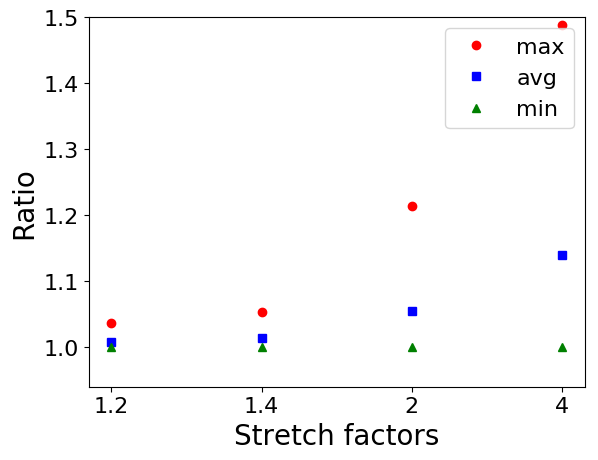}
        \caption{Bottom up}
    \end{subfigure}
    ~
    \begin{subfigure}[b]{0.28\textwidth}
        \includegraphics[width=\textwidth]{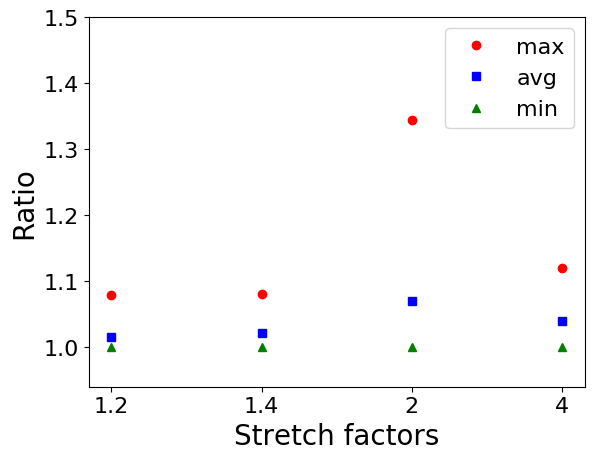}
        \caption{Top down}
    \end{subfigure}
    ~
    \begin{subfigure}[b]{0.28\textwidth}
        \includegraphics[width=\textwidth]{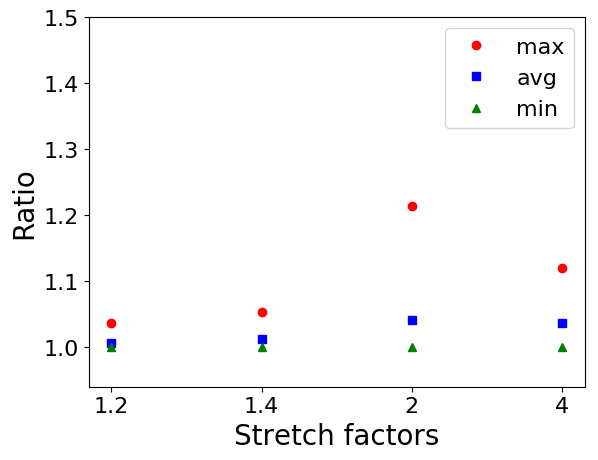}
        \caption{min(BU, TD)}
    \end{subfigure}
    \caption{Performance without oracle on Watts--Strogatz graphs w.r.t.\ the stretch factors} \label{approx_SVR_WS}
\end{figure}

\begin{figure}[htp]
    \centering
    \begin{subfigure}[b]{0.28\textwidth}
        \includegraphics[width=\textwidth]{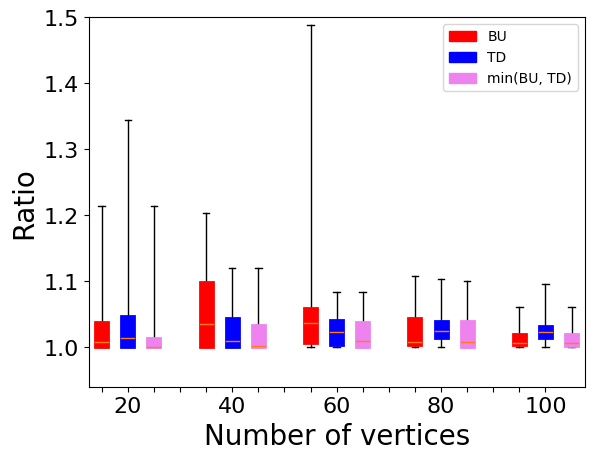}
    \end{subfigure}
    ~
    \begin{subfigure}[b]{0.28\textwidth}
        \includegraphics[width=\textwidth]{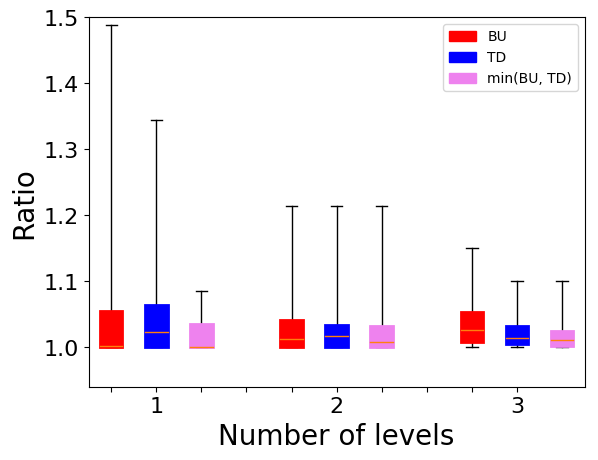}
    \end{subfigure}
    ~
    \begin{subfigure}[b]{0.28\textwidth}
        \includegraphics[width=\textwidth]{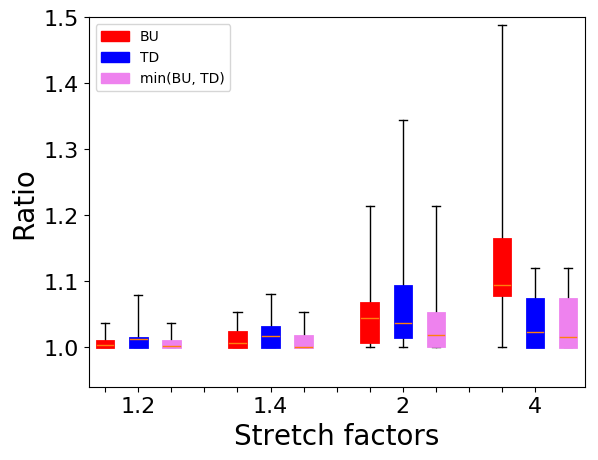}
    \end{subfigure}
    \caption{Performance without oracle on Watts--Strogatz graphs w.r.t.\ the number of vertices, the number of levels, and the stretch factors} \label{approx_box_WS}
\end{figure}

\section{Experimental results on large graphs using Erd\H{o}s-R\'{e}nyi} \label{apdx:ER-large}
Figure \ref{consistent_approx_box_ER} shows a rough measure of performance for the bottom-up and top-down heuristics on large graphs using the Erd\H{o}s-R\'{e}nyi model, where the ratio is defined as the BU or TD cost divided by min(BU, TD). Figure \ref{consistent_time_approx_box_ER} shows the aggregated running times per instance, which significantly worsen as $|V|$ is large.
\begin{figure}[htp]
    \centering
    \begin{subfigure}[b]{0.28\textwidth}
        \includegraphics[width=\textwidth]{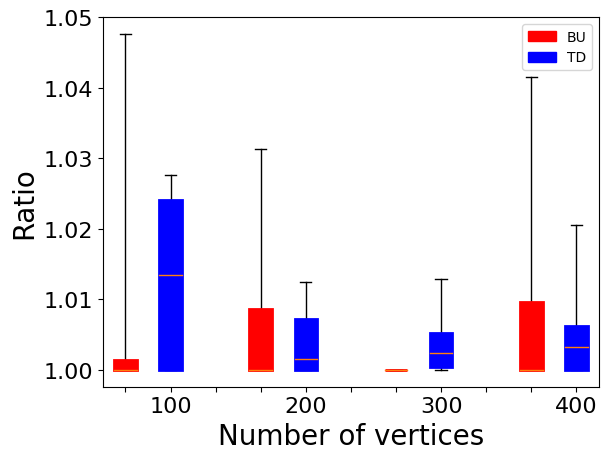}
    \end{subfigure}
    ~
    \begin{subfigure}[b]{0.28\textwidth}
        \includegraphics[width=\textwidth]{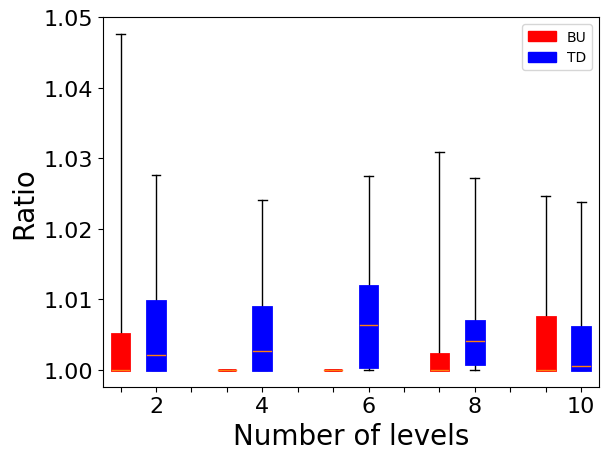}
    \end{subfigure}
    ~
    \begin{subfigure}[b]{0.28\textwidth}
        \includegraphics[width=\textwidth]{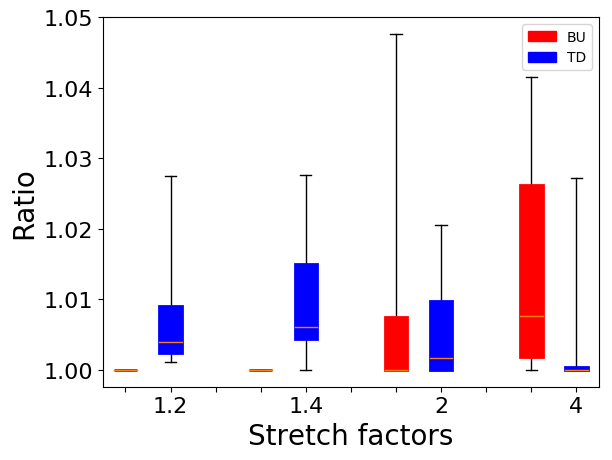}
    \end{subfigure}
    \caption{Performance of heuristic bottom-up and top-down on large Erd\H{o}s--R{\'e}nyi graphs w.r.t.\ the number of vertices, the number of levels, and the stretch factors. The ratio is determined by dividing the objective value of the combined (min(BU, TD)) heuristic.} \label{consistent_approx_box_ER}
\end{figure}

\begin{figure}[htp]
    \centering
    \begin{subfigure}[b]{0.28\textwidth}
        \includegraphics[width=\textwidth]{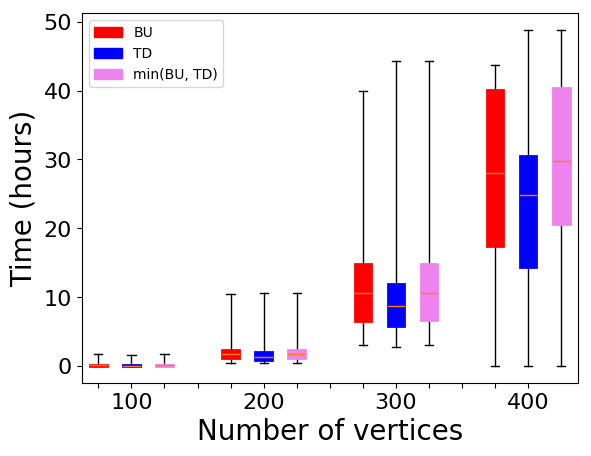}
    \end{subfigure}
    ~
    \begin{subfigure}[b]{0.28\textwidth}
        \includegraphics[width=\textwidth]{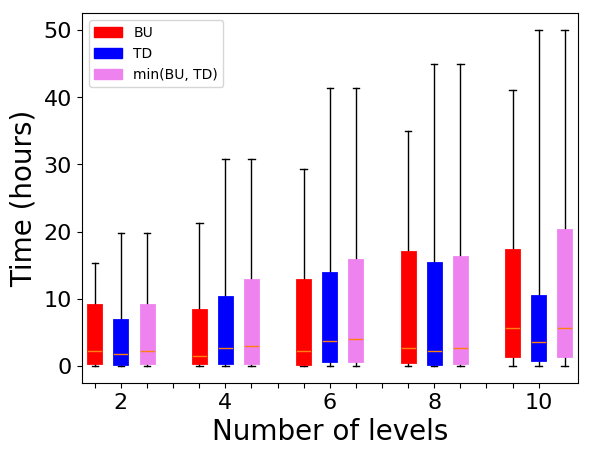}
    \end{subfigure}
    ~
    \begin{subfigure}[b]{0.28\textwidth}
        \includegraphics[width=\textwidth]{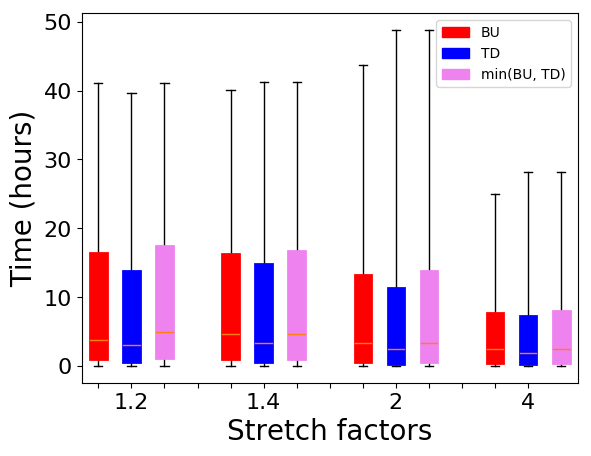}
    \end{subfigure}
    \caption{Experimental running times for computing heuristic bottom-up, top-down and combined solutions on large Erd\H{o}s--R{\'e}nyi graphs w.r.t.\ the number of vertices, the number of levels, and the stretch factors.} \label{consistent_time_approx_box_ER}
\end{figure}
\end{document}